\documentclass{article}
\usepackage{authblk}

\usepackage{makecell}
\usepackage{ulem}
\usepackage{stmaryrd}
\usepackage{comment}
\usepackage{subcaption}
\usepackage{enumitem}
\usepackage{algorithm}
\usepackage{graphicx} 
\usepackage[english]{babel}
\usepackage{amsmath}
\usepackage[noend]{algpseudocode}
\usepackage{graphicx}

\usepackage{pgfplots}
\pgfplotsset{width=10cm,compat=1.9}

\usepackage{url}

\usepackage{amsthm}
\newtheorem{theorem}{Theorem}[section]
\newtheorem{corollary}{Corollary}[section]
\newtheorem{proposition}[theorem]{Proposition}\newtheorem{remark}[theorem]{Remark} 

\usepackage{placeins}
\usepackage{xcolor}
\usepackage{xfrac}
\usepackage{todonotes}
\usepackage{verbatim}
\usepackage{amsfonts}
\usepackage{multirow}
\usepackage{multicol}

\usepackage{xifthen}
\newcommand{\ifequals}[3]{\ifthenelse{\equal{#1}{#2}}{#3}{}}
\newcommand{\case}[2]{#1 #2} 
\newenvironment{switch}[1]{\renewcommand{\case}{\ifequals{#1}}}{}

\AtBeginDocument{\colorlet{defaultcolor}{.}}

\def\version{2}

\newcommand{\add}[1]{\begin{switch}{\version}\case{1}{\textcolor{blue}{#1}}\case{2}{\textcolor{defaultcolor}{#1}}\end{switch}}

\newcommand{\remove}[1]{\begin{switch}{\version}\case{1}{\add{\sout{#1}}}\case{2}{\iffalse{#1}\fi}\end{switch}}

\presetkeys%
    {todonotes}%
    {inline,backgroundcolor=yellow}{}

\begin{document}

\raggedbottom

\title{On the Multi-Commodity Flow with convex objective function: Column-Generation approaches}

\author[1]{Guillaume Beraud-Sudreau}

\author[2]{Lucas Létocart}

\author[1]{Youcef Magnouche}

\author[1]{Sébastien Martin}

\affil[1]{Huawei Technologies Ltd., Paris Research Center, Boulogne-Billancourt, France}
\affil[2]{Université Sorbonne Paris Nord, CNRS, Laboratoire d'Informatique de Paris Nord, LIPN, F-93430 Villetaneuse, France}




\maketitle

\begin{abstract}
The purpose of this work is to develop an algorithmic optimization approach for a capacitated Multi-Commodity flow problem, where the objective is to minimize the total link costs, where the cost of each arc increases convexly with its utilization. This objective is particularly relevant in telecommunication networks, where device performance can deteriorate significantly as the available bandwidth on a link becomes limited. By optimizing this convex function, traffic is efficiently distributed across the network, ensuring optimal use of available resources and preserving capacity for future demands. This paper describes the Convex Multi-Commodity Flow Problem and presents methodologies to solve both its Splittable and Unsplittable variants. In the Splittable version, flows can be fractionally distributed across multiple paths, while in the Unsplittable version, each commodity must be routed through a single path. Our approach employs Column-Generation techniques to address the convexly increasing cost functions associated with arc utilization, effectively accommodating various forms of convex increasing cost functions, including nondifferentiable or black-box convex increasing functions. The proposed methods demonstrate strong computational efficiency, offering a robust framework for managing network flows in complex telecommunication environments.
\end{abstract}

\section{Introduction}

Routing problems in telecommunication networks can be modeled as a Multi-Commodity flow problem \cite{wang2008}. 
The network is represented by a graph $G=(V, A)$ where the nodes are telecommunication devices and the arcs are telecommunication links.
Data transmissions are represented as a set of commodities $K$, where each commodity $k\in K$ requires a given bandwidth $b_k \in \mathbb{R}^+$ between a source node $s_k \in V$ and a target node $t_k \in V$; 

In the convex Multi-commodity flow problem, each arc $a \in A$ has an associated capacity $c_a \in \mathbb{R}^+$ and a convex, lipschitz continuous, increasing cost function $\pmb{r}_a: [0,c_a] \longrightarrow \mathbb R$ of the load flowing through it.

A routing minimizing the sum of arcs costs is relevant in telecommunications for two different reasons.
First, as costs convexity implies an increasing marginal cost of adding new traffic on an already used arc,
a routing strategy minimizing the sum of convex costs over the network will thus avoid creating highly saturated arcs and tend to disperse the traffic along a higher number of moderately used arcs. It ensures that, when possible, each arc will still have some available bandwidth that could be used to accept future demands; this reflects well the needs of telecommunication network operators to ensure that their routing capability will be met in the future and minimize the need of rerouting already accepted demands.
Second, most of the delays of telecommunication network devices relate convexly to their saturation levels; for example, the incurred delay is usually modeled by the Kleinrock function $\sum_a \frac{x_a}{c_a'-x_a}$ where $x_a \in \mathbb{R}^+$ is the flow crossing arc $a \in A$ \add{and $c'_a \in \mathbb{R}^{*}_{+}$}. The Kleinrock function is convex, increasing and lipschitz-contiuous on $[0,c_a]$ for any $c_a \in [0, c'_a[$.

\paragraph{Splittable and Unsplittable Convex Multi-Commoditiy Flow Problem}

The Convex Multi-Commodity Flow Problem (denoted CMCF Problem) requires to define, for each commodity $k \in K$, a flow from $s_k$ to $t_k$ of size $b_k$ such that the sum of all arcs cost functions is minimized.
If a commodity cannot enter the network, a large penalty is applied.
Two Convex Multi-Commodity Flow problems can be considered, depending on the nature of the flows. In the Splittable Convex Multi-Commodity Flow (Splittable-CMCF) problem considered here, each commodity can be routed through a set of paths, each path carrying a fractional part of the commodity.
In the Unsplittable Multi-Commodity Flow problem, each commodity is routed through a single path.

For both the Splittable and Unsplittable-CMCF problems, this paper will focus on the case where the cost functions $\pmb{r}_a$ are increasing, lipschitz-continuous and convex in $[0, c_a]$.
The Splittable-CMCF problem can be formalized as a convex optimization problem, while the Unsplittable-CMCF problem is an integer convex problem. Note that the Splittable version of the problem is a (weak) relaxation of the Unsplittable one.

\paragraph{State of the Art} 
The Multi-Commodity Flow problem (noted MCF), with increasing linear cost objective functions, has been studied for decades. For example, \cite{ford1958} suggested solving the MCF problem with a Column Generation algorithm in 1958.

The more general Convex Multi-Commodity Flow Problem has been studied since 1973, when a first method, the flow deviation method \cite{leblanc1973}, \add{\cite{fratta1973}} was proposed. This approach is based on the Frank-Wolf\add{e} algorithm \cite{frank1956}, and exploits the fact that the subproblem used to generate new search directions in the Splittable-CMCF problem is a simple shortest-path problem, weighted by the marginal cost of each arc at the current solution.

Several other approaches have since improved this idea: for instance, the Projection Method proposed in \cite{bertsekas1983} improves the re-balancing of flows at each iteration of the algorithm. \add{These methods are limited to multi-commodity flow problems without capacity constraints on the arcs.}
Other methods rely on Lagrangian duality to transform the Splittable-CMCF problem into a \add{set of} problem\add{s} of smaller size, which can be solved through a cutting-plane method (e.g. \cite{Goffin1997}).

A survey of Splittable Convex Multi-Commodity Flow methods can be found in \cite{Ouorou2000}.
Unlike the methods described in this work, the survey focuses on a version of the Splittable-CMCF problem \remove{without arc capacities}.
More recently, \cite{fortz2017} proposed a method to solve the Unsplittable-CMCF problem, when the cost functions are increasing and piecewise-linear. This approach is \add{allow to solve the Unsplittable-CMCF problem on instances with \remove{efficient for}} a limited number of segments in the piecewise-linear function. 
In this work, we will consider a more general case where the cost function takes any form, as long as they are increasing convex~: cost functions that are not linear nor differentiable functions, and even non-explicit black-box functions can also be considered. This is a significant generalization compared to the methods described above, which either rely on the derivative of the cost functions or are limited to piecewise-linear functions. This work also generalizes previous convex-MCF methods as it allows integrating constraints such as limited arc capacities.

\paragraph{Contribution}


Our main contribution is to provide different methods for solving the Splittable and Unsplittable-CMCF problems. 
We provide a flexible linearization of the Splittable-CMCF problem, based on an inner-approximation of the cost functions. This linearization can integrate additional constraints and manage non-differentiable (or even black-box) objective cost functions. This linearization is compared to a compact formulation of the problem defined in Section \ref{sect:compact_formulation}, and an extended formulation in Section \ref{sect:convex_formulation}.
Other linearization approaches, via an outer approximation (Kelley's cut \cite{kelley1960}), have been implemented and tested, but their performances were significantly weaker than the inner approximation \add{on the numerical experiments described in \ref{sect:test_description}}~; for this reason, it is not presented in this work.
The three different approaches are compared on SNDLib instances from the literature \cite{orlowski2010}  in Section \ref{sect:perf_split_inner}.


We then consider the Unsplittable-CMCF problem in Section \ref{sect:unsplittable_problem}.
First, additional constrains tightening the Splittable-CMCF Problem are considered in Section \ref{sect:tight_inner_apprx}. Then a tighter relaxation of the Unsplittable-CMCF problem is presented in Section \ref{sect:patterns_problem}. This relaxation is obtained by considering patterns, or sets of commodities, that cross each arc of the network.
A branch-and-price algorithm exploiting these relaxations is described, and the performances are compared in Section \ref{sect:perf_unsplit}.

\section{Splittable-CMCF Problem Modeling: Compact and Extended Formulations}\label{sect:Split_Problem_modeling}

We first describe the Splittable-CMCF problem by providing its compact formulation. As stated above, the Splittable-CMCF problem is both relevant in itself and as a relaxation of the Unsplittable-CMCF problem that will be considered in Section \ref{sect:unsplittable_problem}

\subsection{Compact Formulation} \label{sect:compact_formulation}

The following compact formulation provides a straightforward representation of the Convex Multi-Commodity Flow problem.
Let us denote by $\delta^+(v) \subset A$ (resp. $\delta^-(v) \subset A$) the set of arcs whose heads (resp. tails) are the node $v \in V$. The tail of an arc $a \in A$ is denoted $tail(a)$.
For a commodity $k \in K$ and an arc $a \in A$, let variable $x_k^a \in [0,1]$ be the split ratio of commodity $k$ on arc $a$. Therefore $b_k x_k^a$ \add{(the product of the size of commodity $k$ and the share of this commodity crossing arc $a$)} is the quantity of flow of commodity $k$ crossing arc $a$. Let also variable $y_k \in [0,1]$ be the split ratio of the commodity $k$ not accepted on the network (if the commodity is fully accepted, $y_k=0$).

Two hierarchical objectives are combined: first the demands satisfaction and then the minimization of the sum of arc costs.
Therefore, the variables $y_k$ are weighted with a high value $M b_k$, with $M \in \mathbb{R}^+$ higher than $\sum_{a \in A} L_a$, where $L_a\remove{(c_a)}$ is the \add{L}\remove{l}ipschitz constant of the cost function $\pmb{r}_a$ $\forall a \in A$, to ensure that, for any optimal solution, $b_k y_k$ is minimal. 

Including a $\sum_{k\in K}M b_k y_k$ term in the objective ensures that an optimal solution of the CMCF Problem maximizes the load of accepted commodities in the network: If there exists a feasible solution that accepts a given share $\sigma$ of demands (with $\sigma=\sum_{k \in K} y_k b_k$), any solution that accepts a smaller share of the demands (with $\sum_{k \in K} y_k b_k > \sigma$)  will be suboptimal. 

The compact formulation of the Splittable-CMCF problem, $\mathcal{COMPACT}$, follows the form:
\begin{align}
 & \mathcal{COMPACT}: & \nonumber\\
\qquad & \min  \sum_{a \in A} \pmb{r}_a(\sum_{k \in K} b_k x_k^a  ) + \sum_{k \in K} M b_k y_k & \label{compact:obj}\\
\qquad& \sum_{a \in \delta^+(v)} x_k^a -\sum_{a \in \delta^-(v)} x_k^a = \begin{cases}
1-y_k & \text{if } v = s_k \\
y_k-1 & \text{if } v = t_k \\
0 & \text{otherwise}
\end{cases} & \forall v \in V, k \in K, \label{compact:const:paths_cons}\\ 
\qquad & \sum_{k \in K} b_k x_k^a  \leq  c_a   \quad & \forall a \in A, \label{compact:const:capa}\\
\qquad & x_k^a \geq 0 & \forall a \in A, k \in K,\\
\qquad & y_k \geq 0 & \forall k \in K, 
\end{align}
where the Objective function \eqref{compact:obj} minimizes the sum of all the increasing costs of the arcs, with $\pmb{r}_a$ a convex function for all arc $a \in A$. Constraints \eqref{compact:const:paths_cons} ensure flow conservation at each node, and Constraints \eqref{compact:const:capa} ensure the load of each arc is bounded by its capacity. 

In optimal solutions of the Splittable-CMCF problem, the variables $x_k^a$ are within $[0,1]$ for every commodity $k \in K$ and arc $a \in A$. Indeed, for a given solution, if for a $k \in K$ and an $a \in A$, $x_k^a > 1$, then the same solution with $x_k^a = 1$ is also feasible and has a lower objective value. Similarly, the variables $y_k$ are in $[0,1]$. 
The Unsplittable-CMCF problem restricts variables $x_k^p$ and $y_k$ to $\{0,1\}$. \add{One can also assume, without loss of generality, that no commodities share both their sources and destination. If needed, two such commodities could be replaced during a pre-processing by a single one with the same source and destination, and a bandwidth equal to the sum of the two original commodities' bandwidths.}

Although this compact approach provides a straightforward description of the Splittable-CMCF problem, its dimensionality quickly increases with the problem's size: the number of variables $|A| \times |K|+|K|$ and the number of constraints $|V|\times|K|+|A|$ lead to unpractical representation. 
On the other hand, extended formulation, described in the following sections, relies on a lower number of constraints, and, as variables are explicitly related to paths, they can lead to natural relaxations of the Unsplittable-CMCF problem described in Section \ref{sect:unsplittable_problem}. 
For these reasons, extended formulations of the problem are considered.

\subsection{Extended Formulation and Column Generation} \label{sect:convex_formulation}

The Splittable-CMCF problem can be written in an extended form: For every commodity $k \in K$, the set of paths from $s_k$ to $t_k$ is denoted $P^k$ and, for every path $p \in P^k$, a path variable $x_k^p$ represents the ratio of commodity $k$ using path $p$.

The extended formulation of the Splittable-CMCF problem can be formulated as follows:

\begin{align}
\mathcal{CONVEX}:\qquad  & \min \sum_{a\in A} \pmb 
{r}_a(\sum_{k \in K} \sum_{ p \in P^k | a \in p} b_k x_k^p) + \sum_{k \in K} M b_k y_k \\
\alpha_k:\qquad   & -\sum_{p \in P^k} x_k^p - y_k \add{+ 1 \leq 0} \qquad & \forall k \in K, \label{eq:convex:x_convex}\\
\beta_a:\qquad   &   \sum_{k \in K} \sum_{ p \in P^k | a \in p} b_k x_k^p \add{-c_a \leq 0}  \qquad & \forall a \in A, \label{eq:convex:capa}\\
 \qquad & -x_k^p \leq 0   \qquad & \forall k \in K, p \in P^k, \label{eq:convex:posit_kp}\\
 \qquad &  -y_k \leq 0  \qquad & \forall k \in K. \label{eq:convex:posit_k}
\end{align}
where Constraints \eqref{eq:convex:x_convex} ensure that, for a commodity $k \in K$, the accepted ratio over all paths, together with the non-accepted one, equals at least one. Constraints \eqref{eq:convex:capa} are the capacity constraints, \add{equivalent \remove{similar}} to Constraints \eqref{compact:const:capa} in the compact formulation of the problem. 
Note that this formulation does not rely on the $x_k^a$ variables used in the $\mathcal{COMPACT}$ formulation, but on the path variables.

The variables $\alpha_k$ and $\beta_a$, $\forall k \in K, a \in A$, are the dual variables associated with Constraints \eqref{eq:convex:x_convex} and \eqref{eq:convex:capa}. These dual variables are positive, since all constraints are inequalities \add{(note that this convex problem follows the standard form as defined in \cite{boyd2004}, section 4.1.1)}.

The number of $x_k^p$ variables is exponential, but this formulation can be solved iteratively: at each iteration, a restricted version of the $\mathcal{CONVEX}$ problem (called the Restricted Master Problem or RMP) is solved on a limited set of paths, and new paths to be added for the next iteration are generated via a pricing problem.
Let the variable $x_a=\sum_{k \in K}\sum_{  p \in P^k | a \in p} b_k x_k^p$ be the flow over the arc $a \in A$ and the variable values of an optimal primal-dual solution of the RMP are denoted by "$^*$".

The RMP is similar to the $\mathcal{CONVEX}$ problem above, except that, for all $k \in K$ and $p \in P^k$ with associated variables not included in the restricted problem, the corresponding $x_k^p$ variable is considered to equal $0$. In practice, these variables are not implemented, but in theory they can be considered together with an additional constraint $x_k^p \leq 0$ ; in order to define a pricing problem that will identify which variables $x_k^p$ must be added to the RMP, a dual variable $\epsilon_k^p$ is attached to the constraint $x_k^p \leq 0$ for all $k \in K$ and $p \in P^k$ such as $x_k^p$ has not been generated yet.

At an optimal primal-dual solution of the RMP, the Karush-Kuhn-Tucker optimality conditions ("KKT conditions") must hold true \cite{boyd2004}. In particular, the condition derived from the variable $x_k^p$ ensures that $\sum_{a \in p} b_k \pmb{r}'_a(x^*_a) - \alpha_k^* + \sum_{a \in p} \beta_a^* b_k + {\epsilon_k^p}^*= 0$, where $\pmb{r}'_a(x^*_a)$ is the derivative of $\pmb{r}_a(x^*_a)$ for all $a \in A$.

If, for some $k \in K$ and $p \in P^k$, ${\epsilon_k^p}^* = -\sum_{a \in p} b_k \pmb{r}'_a(x^*_a) + \alpha_k^* - \sum_{a \in p} \beta_a^* b_k \add{>} 0$, the constraint $x_k^p \leq 0$ is active, and the variable $x_k^p$, which is potentially non-null, needs to be introduced into the RMP.

This allows us to formulate the pricing problem as the following problem \add{to find paths associated with potentially non-null variables}:
\begin{equation*}
\text{find a path $p \in P^k$ with cost} \quad \sum_{a \in p} (\beta_a^* + \pmb{r}'_a(x^*_a)) < \frac{\alpha_k^*}{b_k}
\end{equation*}

This pricing problem can be easily solved in polynomial time, for every demand $k \in K$, by finding the shortest path from $s_k$ to $t_k$ in the graph $G = (V,A)$ where each link is weighted by $\beta_a^* + \pmb{r}'_a(x^*_a)$.
Note that this approach is close to the Simplicial Decomposition method presented in \cite{vonhohenbalken1975}, except for the fact that the pricing problem is derived from the KKT optimality conditions and not from a step of the Frank\remove{e}-Wolfe algorithm. This difference allows for greater flexibility, for instance \add{allowing non-differentiable cost functions or} integrating constraints such as the capacity constraint \eqref{eq:convex:capa} as weights in the shortest-path problem. 

\subsection{Problems Complexity} \label{sect:complexity_split}

The Splittable Linear MCF, similar to the Splittable-CMCF restricted to linear function, is well known to be a linear program and thus can be solved in polynomial time.

\add{ Both the $\mathcal{COMPACT}$ and the $\mathcal{CONVEX}$  formulations require one to solve a convex optimization problem, either directly the $\mathcal{COMPACT}$ problem or the Restricted Master Problem for the $\mathcal{CONVEX}$ Column Generation. 
Interior-point methods converge to the optimal solution of the Splittable-CMCF problem with a given precision in polynomial time, provided functions $\pmb{r}_a$ and its derivative $\pmb{r}'_a$, for all $ a \in A$. And as the pricing of the problem of the $\mathcal{CONVEX}$ formulation can be solved in polynomial time, the separation/optimization equivalence \cite{grotschel1993} allows concluding that both problems are polynomial for a given precision.}

\add{So one can conclude that, like the Splittable Linear MCF problem, both the $\mathcal{COMPACT}$ and the $\mathcal{CONVEX}$ problems can be solved in polynomial time.
}

\subsection{Numerical Tests description} \label{sect:test_description}

The formulations of the Splittabe-CMCF problem $\mathcal{COMPACT}$ and $\mathcal{CONVEX}$ have been implemented in C++ using CPLEX version 22.1.1 \add{; the Flow-Deviation algorithm has also been implemented in C++. All tests were run \remove{running}} on 2 cores with 64GB of RAM.
These implementations have been tested on instances of the SND library \cite{orlowski2010}. 
In all instances, for all links $(u,v) \in A$, a new arc $(v,u)$ is created with the same capacity as $(u,v)$ if it does not exist in the graph.
Each instance is scaled to ensure that its capacities are slightly above congestion for the Splittable-CMCF problem: first, a factor $\tau \in \mathbb{R}^+$, unique for each instance, is computed such that an instance with all capacities scaled by $\tau$ is feasible with all commodities accepted, i.e. $y_k=0$ for all $k \in K$, while an instance with capacities scaled by $\tau \times 0.99$ has no feasible solution with all commodities accepted. This factor $\tau$ can be computed by dichotomy, solving a sequence of Splittable Linear MCF problems (which are simple LPs, as seen in Section \ref{sect:complexity_split}). 

In order to ensure that multiple solutions that accept all commodities exist, the capacities computed above are then increased by a factor of 1.05 to create the instances used in the numerical experiments.

Two cost functions were considered: 
\begin{itemize}[left= 20pt, topsep=3pt]
    \item a Kleinrock cost function $\pmb{r}_a(x_a)=f_a/(d_a - x_a)$, where $d_a=1.01 c_a$ to ensure finite costs when $x_a \in [0,c_a]$,
    \item a quadratic cost function $\pmb{r}_a(x_a)=f_a \times x_a^2$.
\end{itemize} 

For all arcs $a \in A$, the factors $f_a$ were calculated to ensure that $\pmb{r}_a(c_a)$ matched the arc costs when provided by the libraries: this allows creating instances with a different cost on each arc. For instances in which the SND library does not provide any cost or capacities, a constant value is used on all arcs.  Figure \ref{fig:cost_functions} represents these functions.
\add{Constraints linking the cost of each arc with its load were expressed either as Second-Order cone constraints - for the Kleinrock cost function - or as quadratic constraints - for the Quadratic cost function. As the Flow-Deviation algorithm cannot solve capacitated multi-commodity flow problems, the capacity constraints are removed when testing this particular algorithm and, for tests using the Kleinrock cost function, a modified objective function similar to the one presented in Figure 1 of \cite{Ouorou2000} was used. These change did not  impact the optimal solution of tests on instances on the Kleinrock cost function, but had a very significant solution on tests on instances on the Quadratic loss function - where optimal results often displayed loads above the instances' arcs capacities.}

\begin{figure}[ht]
    \centering
    \includegraphics[width=.6\textwidth]{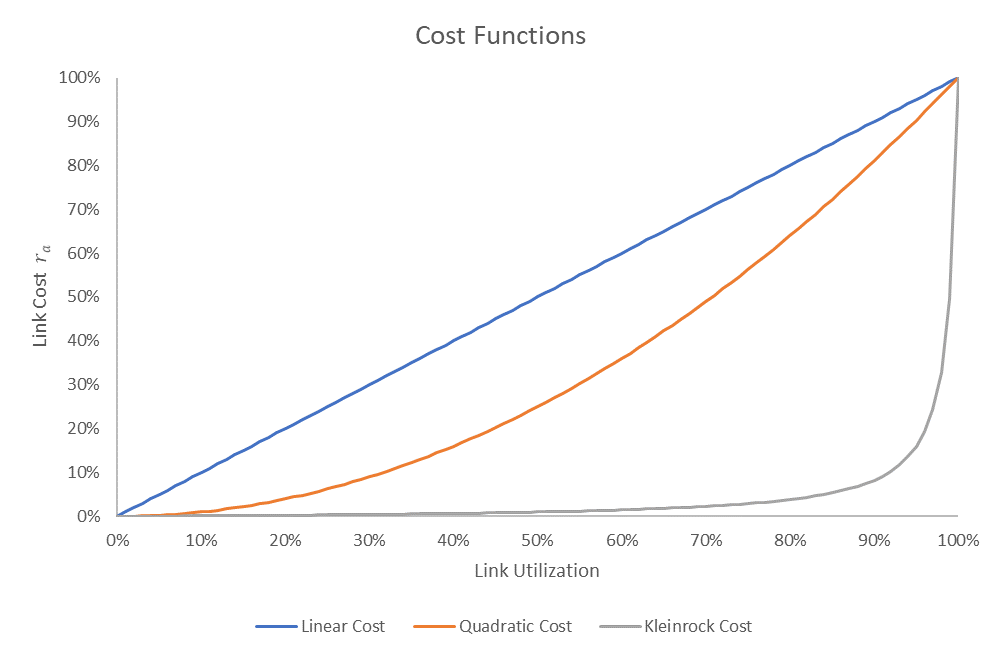}
    \caption{Cost function used in numerical experiments: Linear, Quadratic and Kleinrock costs. }
    \label{fig:cost_functions}
\end{figure}

Figure \ref{fig:arcs_usage_cumul} shows the cumulated distributions of the arcs available capacities $c_a - x_a$ for an optimal solution of the ZIB54 instance, optimized with 3 different cost functions. The routing minimizing convex costs avoids high usage of arcs, except when this high usage is necessary to ensure the feasibility of the problem. In this example, over $10\%$ of the arcs are fully used (i.e. they have available capacity of $0\%$) in a solution minimizing linear costs, compared to less than $3\%$ when minimizing the quadratic costs and zero when minimizing the Kleinrock function.

This observation is consistent with the marginal costs of the tested arc cost functions. When using linear functions, the marginal cost is constant, while it increases with the arc utilization for quadratic cost functions and increases very steeply for high utilization for Kleinrock functions.

Likewise, quadratic cost functions have a marginal cost equal to zero when the link utilization is null (or the available capacity equals $100\%$). For this reason, one can observe that in the instance described in Figure \ref{fig:arcs_usage_cumul} no link has an available capacity of $100\%$ for quadratic cost functions, while they do for Kleinrock or linear cost functions.

\begin{figure}[ht]
    \centering
    \includegraphics[width=.6\textwidth]{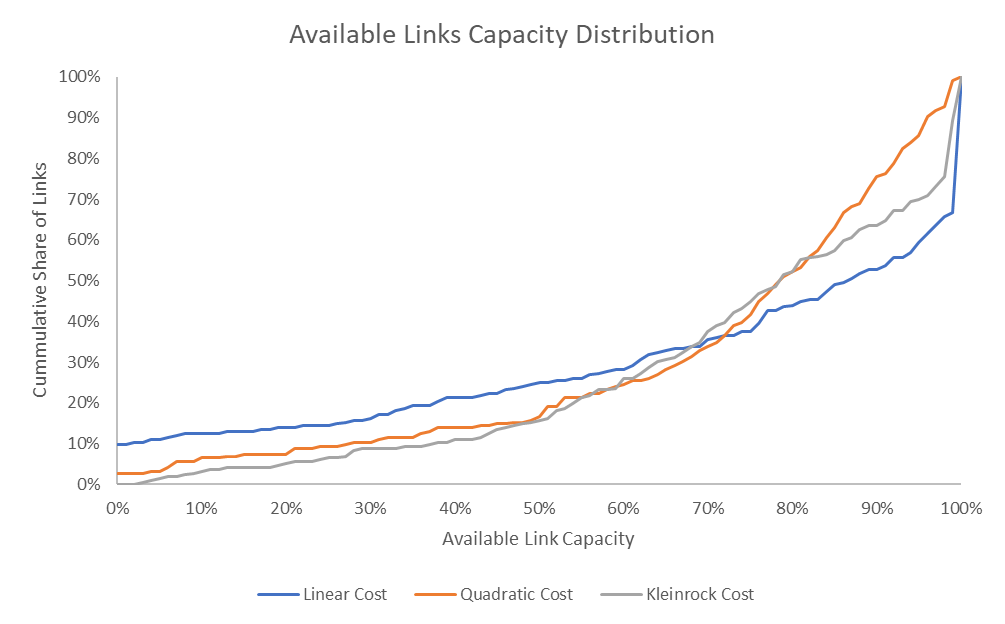}
    \caption{Example of cumulated distribution of available capacity of arcs on the SNDLib instance ZIB54 \add{(|K=1501|, |A|=160)}, with optimal routing minimizing linear, quadratic or Kleinrock cost functions.
    }
    \label{fig:arcs_usage_cumul}
\end{figure}

\subsection{Performances} \label{sect:perf_split_compact_convex}

Performance charts of the optimization approaches described in Sections \ref{sect:compact_formulation} and \ref{sect:convex_formulation}, showing the share of SNDLib instances solved to a precision of $0.1\%$ in a given time, are presented in Figure \ref{fig:perf_chart_compconv}. 
The $\mathcal{CONVEX}$ column generation method is shown to be significantly more efficient than direct optimization of the $\mathcal{COMPACT}$ formulation on all instances. Note that the durations are represented with a logarithmic scale, so the differences between the two methods are very significant, with factors from $10$ to $40$ for all the larger instances. This gain seems to be on the same scale for quadratic or Kleinrock cost functions, and grow with the size of the instance: $\mathcal{COMPACT}$ formulation can't be solved for the larger instances due to memory limits. \add{Finally, the Flow-Deviation algorithm proved less efficient than the compact approach when tested on instances with the Kleinrock function; on tests with quadratic cost functions, performances comparisons were not relevant as the optimization problem did not include arc capacities and led to optimal solutions widely different from the Splittable-CMCF problem tested with the other approaches.}

The Column Generation approach performances are, however, still limited by the need to solve, for each iteration, a convex, nonlinear, Restricted Master Problem (RMP), which can be computationally expensive. This effect is limited for quadratic (resp.  Kleinrock) cost functions as the RMP can be expressed as a quadratic program (resp. a Second-Order Cone Program), but could impact the $\mathcal{CONVEX}$ column generation performances if such formulations did not exist. Furthermore, the pricing problem requires computing the gradient $\pmb{r}'_a$ of the cost functions $\pmb{r}_a$ for all arcs $a \in A$, which may not be defined.
For these reasons, in the following sections, we will consider other column generation approaches based on linearization of the problem.

\begin{figure}[ht]
    \centering
    \begin{subfigure}{.45\textwidth}
    \includegraphics[width=1\linewidth]{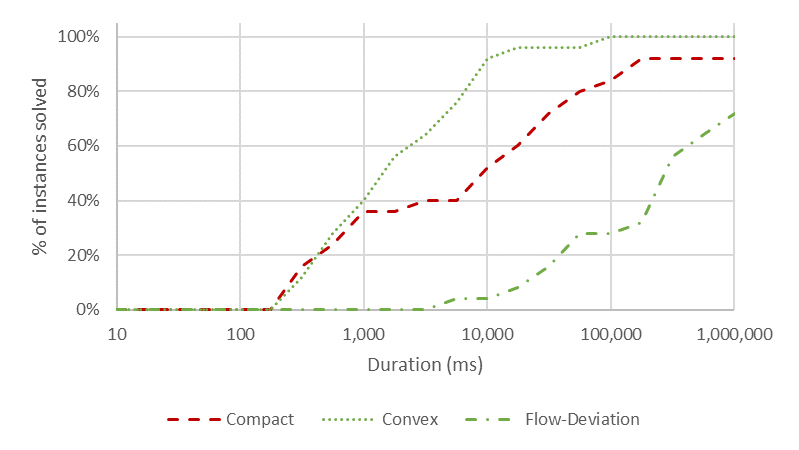}
    \caption{Kleinrock cost functions}
    \end{subfigure}
    \begin{subfigure}{.45\textwidth}
    \includegraphics[width=1\linewidth]{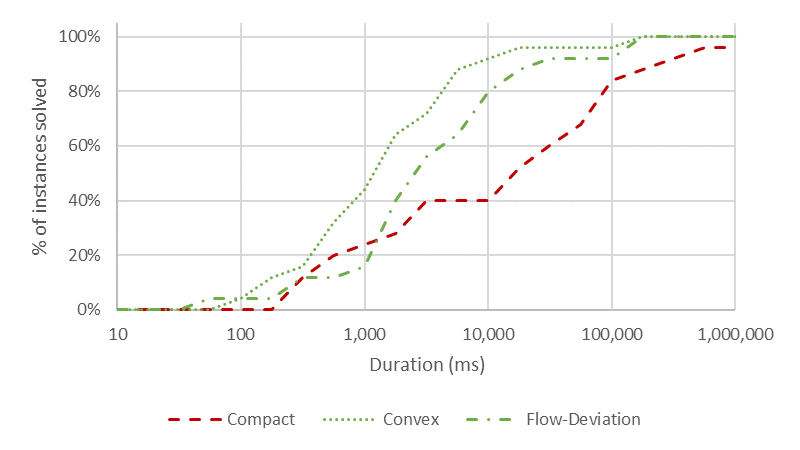}
    \caption{Quadratic cost functions}
    \end{subfigure}
    \caption{Performance Chart of the Splittable-CMCF problem \add{on instances of the SNDLib} with 2 cost functions \add{; the problem tested for the Flow-Deviation algorithm is a relaxation of the Splittable-CMCF problem with uncapacitated arcs.}}
    \label{fig:perf_chart_compconv}
\end{figure}

\FloatBarrier
\section{Splittable-CMCF Problem:  Inner Approximation Formulation}
\label{sect:inner_approx}

The decomposition method described above requires to solve, at each iteration, a nonlinear Restricted Master Problem (RMP), which can be time consuming for large instances. Furthermore, this approach relies on the KKT optimality conditions to derive the pricing problem, which means that the cost functions need to be differentiable. 

An outer-approximation approach, where the cost functions are linearized by a set of constraints generated by a cutting plane algorithm, has been considered. However, this method simultaneously generates a large number of path variables and cuts, leading to very large constraint matrices. In practice, it leads to a very large restricted master problem, and proves inefficient to solve the Splittable-CMCF problem.

\subsection{The Inner Approximation Formulation}

For this reason, it is relevant to consider solving the problem by approximating, for every arc $a \in A$, the cost functions $\pmb{r}_a$ with polytopes of vertices $(c_a^i, \pmb{r}_a(c_a^i))$, with index $i \in N=\{1,...,n\}$ \add{which can be generated through a Column-Generation scheme, as described below}. Each point $(x_a, \pmb{r}_a(x_a))$ can then be approximated by a convex combination of the vertices of the polytopes, weighted by a new set of variables $z_a^i \in \mathbb R^+$. As, for a given arc $a \in A$, the set of variables $z_a^i$ are the weights of a convex combination, we have $\sum_{i \in N} z_a^i=1$, $\sum_{i \in N} z_a^i c_a^i$ representing the usage of arc $a$ and $\sum_{i \in N} z_a^i \pmb{r}_a(c_a^i)$ the cost associated with arc $a$.

Figure \ref{fig:epigraph_approx} provides an illustration of the Inner-Approximation, where a cost function's epigraph (in light green) is approximated by a 4-vertices polyhedron (in red - vertices $c_a^i$ are shown as red crosses). The possible values of $(x, \pmb{r}_a(x))$ of the epigraph are approximated as convex combinations of vertices $\{(c_a^i, \pmb{r}_a(c_a^i))\}$ of the polyhedron. 
New vertices $\{(c_a^i, \pmb{r}_a(c_a^i))\}$ of the polyhedron are added iteratively through a pricing problem, leading to a closer approximation of the cost function $\pmb{r}_a$.

Note that the polyhedral approximation does not exactly match the cost function's epigraph. Points above the line $(c_a^0, \pmb{r}_a(c_a^0))-(c_a^3, \pmb{r}_a(c_a^3))$ are necessarily suboptimal for the Objective \eqref{eq:inner_obj}. And, even when the column generation process converges to an optimal solution, some points $(x, \pmb{r}_a(x))$ are not in the polyhedron.

\begin{figure}[ht]
    \centering
    \includegraphics[trim={0 3.25cm 0 2.5cm},clip,width=0.6\linewidth]{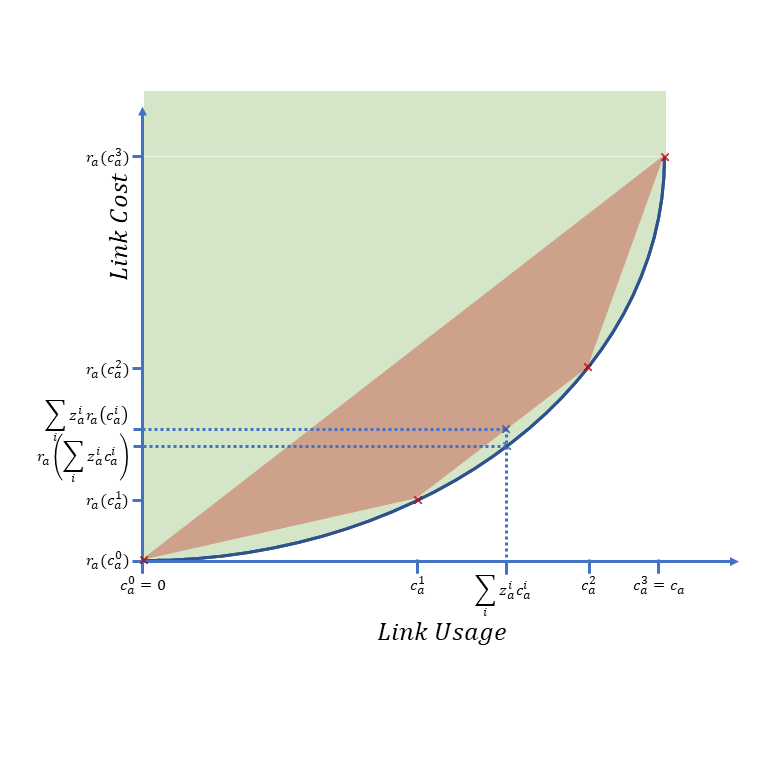}
    \caption{Example of cost function $r_a$ and an inner-approximation of its epigraph.}
    \label{fig:epigraph_approx}
\end{figure}

This approximation effectively linearize the cost functions $\pmb{r}_a$ by generating new points $(c_a^i, \pmb{r}_a(c_a^i))$ and associated variables $z_a^i$ .
This inner approximation formulation of the Splittable-CMCF problem, $\mathcal{INNER}$, can be expressed as:
\begin{align}
\mathcal{INNER}: \qquad &    \min \sum_{a \in A} \sum_{i \in N}  \pmb{r}_a(c_a^i) z_a^i + \sum_{k \in K} M b_k y_k \label{eq:inner_obj}\\
\alpha_k : \qquad & - \sum_{p \in P^k} x_k^p - y_k \leq -1 \qquad & \forall k \in K,  \label{eq:inner_demand_convexity}\\
\beta_a : \qquad  &  \sum_{k \in K} \sum_{ p \in P^k | a \in p} b_k x_k^p - \sum_{i \in N} c_a^i z_a^i    \leq 0 \qquad &\forall a  \in A, \label{eq:inner_usage_arc}\\
\gamma_a :\qquad & \sum_{i \in N} z_a^i = 1 \qquad & \forall a \in A, \label{eq:inner_polytop_convexity}\\
    & x_k^p \geq 0, y_k \geq 0 \qquad & \forall k \in K,p \in P^k, \label{eq:inner_x_y_def}\\
    & z_a^i \geq 0 \qquad & \forall a \in A, i \in N,  \label{eq:inner_z_def}
\end{align}
where the objective function \eqref{eq:inner_obj} is represented as a sum of convex combinations of costs ; note that the scalars $c_a^i$ are constant, so the $\mathcal{INNER}$ problem is linear. 
Constraints \eqref{eq:inner_demand_convexity} ensure that demands are accepted whenever possible, Constraints \eqref{eq:inner_polytop_convexity} define variables $z_a^i$ as a convex combination of polytopes vertices, and Constraints \eqref{eq:inner_usage_arc} link the path usage variables $x_k^p$ to the convex representation of the usages of the arcs $\sum_{i \in N} c_a^i z_a^i$.

Variables $\alpha_k$, $\beta_a$ and $\gamma_a$ are the dual variables associated with constraints \eqref{eq:inner_demand_convexity}, \eqref{eq:inner_usage_arc} and \eqref{eq:inner_polytop_convexity} respectively.
The dual problem associated with $\mathcal{INNER}$ allows us to formulate two pricing problems. Indeed, new paths and new inner-approximation vertices are dynamically generated from an optimal solution of the Restricted Master Problem (denoted by $^*$) by solving the pricing problems:
\begin{align}
    \text{Find new path $p$ in $P^k$ such as: } \quad & \sum_{a \in p} \beta_a^* < \sfrac{\alpha_k^*}{b_k}  \\
    \text{Find $c_a^i \in [0,c_a]$ such as: } \quad&    \gamma_a^* < \beta_a^* c_a^i - \pmb{r}_a(c_a^i) \label{eq:inner_vertices_pricing}
\end{align}
This approach requires to generate dynamically two distinct sets of columns, one set related to the paths and one other related to the vertices of the cost functions. 
The pricing problem used to generate new path variables is a shortest path problem, very similar to the one solved in the $\mathcal{CONVEX}$ formulation, with arcs weighted by the optimal values of the $\beta_a$ variables in the RMP.
Likewise, new vertices of the inner approximation can be generated through a simple one-dimensional convex optimization problem. If derivatives of the arc cost functions are available, the pricing problem can be solved explicitly. For example, for quadratic costs $\pmb{r}_a(x_a)=f_a  x_a^2$, a new vertex $c_a^i = \frac{\beta_a^*}{2 f_a}$ is generated if $\gamma_a^* < \beta_a^* c_a^i - f_a {c_a^i}^2 = \frac{{\beta_a^*}^2}{4 f_a}$ ; for a Kleinrock function $\pmb{r}_a(x_a)=f_a/(d_a - x_a)$, a new vertex $c_a^i = d_a - \sqrt{\frac{f_a}{\beta_a^*}}$ is generated if $\beta_a^* d_a - 2 \sqrt{\beta_a^* f_a} > \gamma_a^*$.
If the derivative of the cost functions is not explicit, the problem can still be solved by a dichotomic search. 

For the $\mathcal{COMPACT}$ and the $\mathcal{CONVEX}$ problems of the previous section, it can be easily shown that the $\mathcal{INNER}$ problem can be solved with a given precision in polynomial time with the proposed column generation approach.
Despite the higher number of variables, the $\mathcal{INNER}$ Column Generation approach is more efficient than the $\mathcal{CONVEX}$ Column Generation algorithm as the Restricted Master Problem is linear, and more flexible as it allows the use of non-differentiable cost functions.

The formulation $\mathcal{INNER}$ only computes the values of the cost function in the pricing problem of the vertices, which can be solved by computing a few values of the cost function. This means that the complexity of computing the values of the pricing functions does not impact the global cost of solving the problem.
\add{As this formulation does not rely on the cost function's derivatives, it can be directly applied to non-differentiable function as long as it can be evaluated and a solver of the pricing problem is provided. Figure \ref{fig:epigraph_nondif} provides an example of such a function - for this instance cost functions of the shape $\pmb{r}_a(x_a)=f_a \max((1-\cos(x_a/c_a)), (x_a/c_a)^2-0.1, 2 e^{x_a/c_a}-4)$ were used. Note this function is convex in $[0,c_a]$ and non-differentiable in $2$ points.}

\begin{figure}[ht]
    \centering
    \includegraphics[trim={0 0.4cm 0 0.2cm},clip,width=0.6\linewidth]{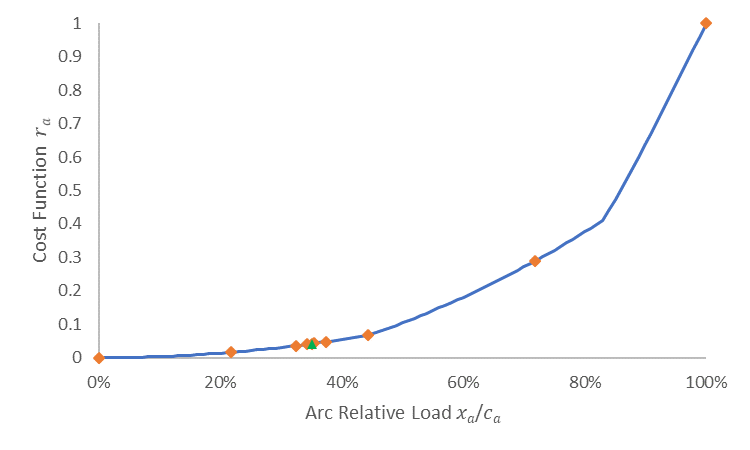}
    \caption{\add{Example of convex, increasing, non-differentiable cost function, in blue, and vertices of a possible inner-approximation in orange.}}
    \label{fig:epigraph_nondif}
\end{figure}

The following remark can be derived from the convexity of the $\pmb{r}_a$ functions:
\begin{remark}
    There always exists an optimal solution of the $\mathcal{INNER}$ problem with, for every arc $a \in A$, a single variable $z_a^i = 1$.
    Furthermore, if the function $\pmb{r}_a$ is strongly convex, then for any optimal solution of $\mathcal{INNER}$ and the arc $a \in A$, a single variable of $\{z_a^i\}_{i \in N}$ is non-null for any optimal solution \add{, as strong convexity would imply that a convex combination of two or more vertices $(x_a, \pmb{r}_a(x_a))$ is strictly above the function $\pmb{r}_a$. }
\end{remark}

\subsection{Extension to non-convex cost functions}

When applied to any continuously increasing, potentially nonconvex, cost function $\{\pmb{r}_a\}_{a \in A}$, the optimization approach described in Section \ref{sect:inner_approx} converges to the solution of a CMCF problem, where the cost functions are the convex \add{hulls \remove{envelopes}} of the used cost functions.

\begin{proposition}
    The Column Generation algorithm associated to the $\mathcal{INNER}$ problem applied to a given graph $G=(V,A)$ and a set of commodities $K$ with continuously increasing non-convex functions $\{\pmb{r}_a\}_{a \in A}$, or for their convex \add{hulls \remove{ envelopes}} $\{env(\pmb{r}_a)\}_{a \in A}$ converges to the same solution.
\end{proposition}

\begin{proof}
    It is sufficient to prove that for a given arc $a$ and dual values $\gamma_a^*$ and $\beta_a^*$, the Pricing Problem \eqref{eq:inner_vertices_pricing} provides the same solution $c_a^i$ with a cost function $\pmb{r}_a$ or $env(\pmb{r}_a)$.

    Observe that, for any continuous function $f: [l, r] \longrightarrow \mathbb R$ with $(l,r)\in \mathbb{R}^2$ and $l < r$, we have $\min f = \min~env(f)$ and $x = \max~argmin~f = \max~argmin~env(f)$.  Taking $f(c)=-\beta_a^* c + \pmb{r}_a(c)$ for $c \in [0,c_a]$, we have $env(f)(c)=-\beta_a^* c+ env(\pmb{r}_a)(c)$. Applying the Pricing Problem \eqref{eq:inner_vertices_pricing} to a cost function $\pmb{r}_a$ returns $\max~argmax~ \beta_a^* c - \pmb{r}_a(c)= \max~argmin~ -\beta_a^* c+ \pmb{r}_a(c)= \max~argmin~ -\beta_a^* c+env(\pmb{r}_a)(c)$, which is the same value as applying the Pricing Problem \eqref{eq:inner_vertices_pricing} to a cost function $env(\pmb{r}_a)$.
\unskip\nobreak\hfill $\square$
\end{proof}
Figure \ref{fig:epigraph_nonconvex} provides an illustration of such an approximation, where a non-convex cost function (in blue) is approximated \add{by a polyhedron with vertices displayed as orange diamonds. The value of the optimal solution is displayed as a green triangle. Notice that only a limited number of vertices, all on the frontier of the convex envelope of the function's epigraph, have been generated by the pricing problem \eqref{eq:inner_vertices_pricing}. For instance, for this particular arc, no vertex has been generated between points $(0,0)$ and $(0.55, \pmb{r}_a(0.55 c_a))$ The cost-function used in the instance displayed was $\pmb{r}_a(x_a)=x_a/c_a + 2 (x_a/c_a)^2 + \sqrt{x_a/c_a} \sin(20 x_a/c_a) /4$.
\remove{by its convex  envelope (in orange) ; the inner-approxmation polyedron generated by the pricing problem \eqref{eq:inner_vertices_pricing} is displayed in light red.}
}

\begin{figure}[ht]
    \centering
    \includegraphics[trim={0 0.4cm 0 0.2cm},clip,width=0.6\linewidth]{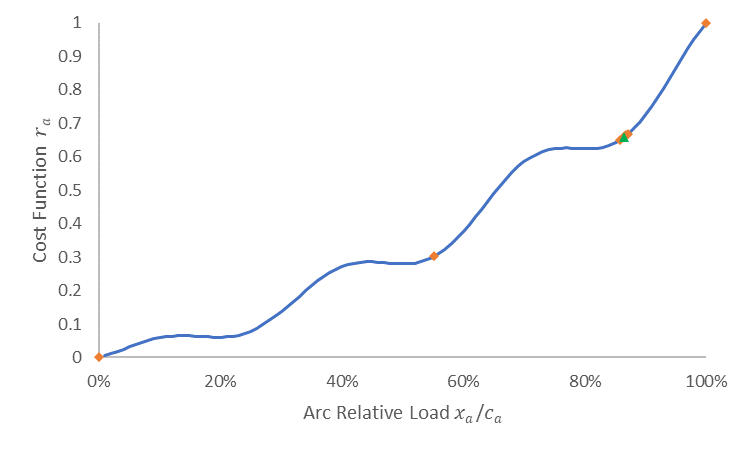}
    \caption{Example of non-convex cost function with an increasing convex hull, in blue, and vertices of a possible inner-approximation in orange.}
    \label{fig:epigraph_nonconvex}
\end{figure}

\FloatBarrier

\subsection{Performances}\label{sect:perf_split_inner}

Figure \ref{fig:relaxed_perfs} shows performance graphs of the 3 different formulations for the SNDLib instances with quadratic and Kleinrock cost functions. 
As we saw in Section \ref{sect:perf_split_compact_convex}, directly solving the $\mathcal{COMPACT}$ formulation is less efficient than the $\mathcal{CONVEX}$ column generation approach. It also appears that the $\mathcal{INNER}$ formulation, which linearizes the optimization problem, is significantly faster than the $\mathcal{CONVEX}$ formulation that directly handles convex optimization problems.
This behaviour is true for all instances, with the solving times divided by $2$ on most instances (note that the time follows a logarithmic scale in the performance graphs).

This performance gain is obtained by avoiding the resolution of convex optimization problems for the Restricted Master Problems in the $\mathcal{INNER}$ approach. As paths and cost-approximation vertices are generated within the same iterations in the $\mathcal{INNER}$ column generation scheme, the linearization of the Restricted Master Problem does not lead to a significantly higher number of Restrict Master Problem resolutions. The slower solving of a convex RMP consequently leads to a longer overall time for the $\mathcal{CONVEX}$ formulation than for the $\mathcal{INNER}$ one.
These higher performances, combined by its flexibility (as non-derivable convex cost functions can be integrated in the problem with no efforts) motivate the use of the inner approximation column generation for the Splittable-CMCF problem.

\begin{figure}[ht]
    \centering
    \begin{subfigure}{.45\textwidth}
    \includegraphics[width=1\linewidth]{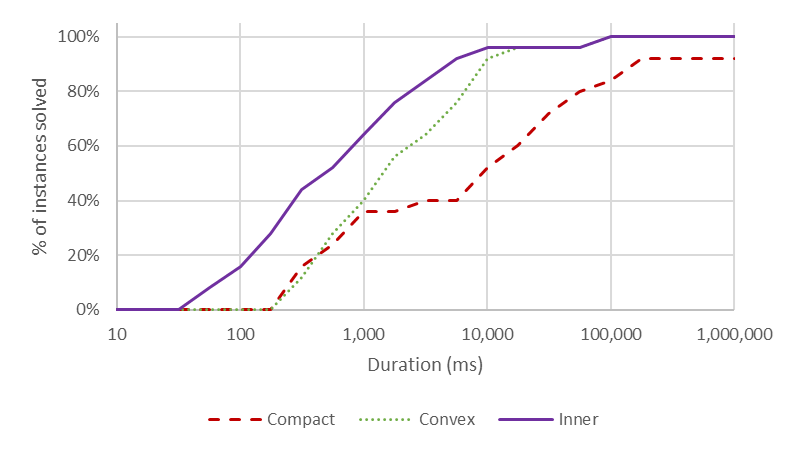}
    \caption{Kleinrock cost functions}
    \end{subfigure}
    \begin{subfigure}{.45\textwidth}
    \includegraphics[width=1\linewidth]{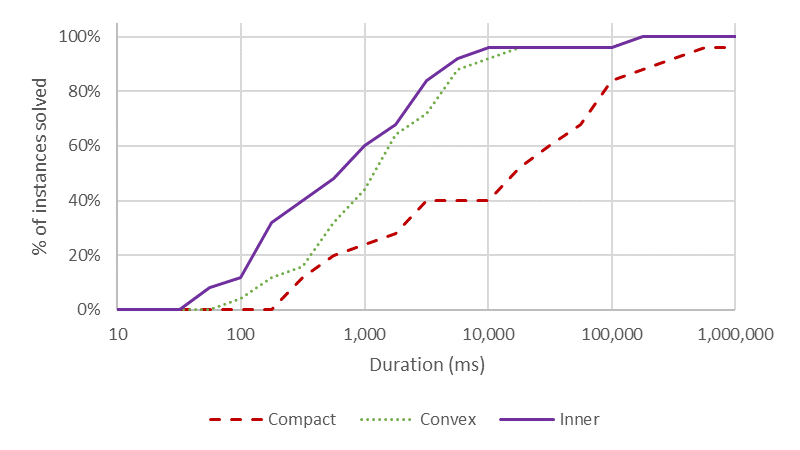}
    \caption{Quadratic cost functions}
    \end{subfigure}
    \caption{Performance Chart of the Splittable-CMCF problem with two cost functions, comparing the $\mathcal{COMPACT}$, $\mathcal{CONVEX}$ and $\mathcal{INNER}$ approaches.}
    \label{fig:relaxed_perfs}
\end{figure}

\FloatBarrier

\section{Unsplittable-CMCF Problem}\label{sect:unsplittable_problem}

The Unsplittable Convex Multi-Commodity Flow problem (Unsplittable-CMCF problem) is similar to the Splittable-CMCF problem described in Section \ref{sect:Split_Problem_modeling}, except that each commodity $k \in K$ must follow a single path from its source $s_k \in V$ to its target $t_k \in V$. Furthermore, if a commodity is accepted in the network, it must be fully accepted.

\subsection{Problem Description}

Unsplittable versions of the $\mathcal{COMPACT}$, $\mathcal{CONVEX}$ and $\mathcal{INNER}$ problems can be defined as follows:
\begin{align*}
 \mathcal{U-COMPACT}: \qquad & \min \sum_{a \in A} \pmb{r}_a(\sum_{k \in K} x_k^a b_k ) + \sum_{k \in K} M b_k y_k \\
&\eqref{compact:const:paths_cons}, \quad  \eqref{compact:const:capa}, \\
& x_k^a \in \{0,1\}  \quad \forall a \in A, k \in K, \qquad y_k \in \{0,1\} \quad   \forall  k \in K.
\end{align*}
\begin{align*}
\mathcal{U-CONVEX}: \qquad  & \min_{x_k^p, y_k } \sum_{a\in A} \pmb {r}_a(\sum_{k \in K} \sum_{ p \in P^k | a \in p} b_k x_k^p) + \sum_{k \in K} M b_k y_k \\
& \eqref{eq:convex:x_convex}, \quad \eqref{eq:convex:capa}, \\
& x_k^p \in \{0,1\}  \quad  \forall k \in K, p \in P^k, \qquad
  y_k \in \{0,1\}  \quad \forall k \in K. 
  \end{align*}
\begin{align*}
\mathcal{U-INNER}: \qquad  & \min \sum_{a \in A} \sum_{i \in N} z_a^i \pmb{r}_a(c_a^i) + \sum_{k \in K} M b_k y_k \\
& \eqref{eq:inner_demand_convexity}, \quad \eqref{eq:inner_polytop_convexity}, \quad \eqref{eq:inner_usage_arc}, \quad \eqref{eq:inner_z_def},\\
& x_k^p \in \{0,1\} \quad  \forall k \in K, p \in P^k, \qquad
  y_k \in \{0,1\} \quad \forall k \in K. 
  \end{align*}
\begin{remark} \label{rem:unique_var_at_one}
For any optimal solution of $\mathcal{U-CONVEX}$, exactly one of the variables $\{{y_k}\} \cup \{{x_k^p}\}_{p \in P^k} $ equals one, while the others equal zero.
\end{remark}
This remark allows to notice that optimal solutions of the formulations above are indeed unsplittable flows between the commodities sources and destinations, whenever such a solution exists; which motivates the Unsplittable-CMCF problem name. It also holds for the Unsplittable version of the problem $\mathcal{U-INNER}$.

\subsection{Complexity} \label{sect:complexity_unsplit}

The Unsplittable-CMCF problem is $\mathcal{NP}-hard$ ; for example, \cite{kleinberg1996} shows that the bin-packing problem can be reduced to an Unsplittable Multi-Commodity Flow problem (with linear costs) on a graph with only 2 nodes, proving the complexity of the Unsplittable Linear MCF problem, which is a special case of the Unsplittable-CMCF problem.

However, the Unsplittable Convex Single Commodity Flow problem, which is a special case of the Unsplittable-CMCF problem where $|K|=1$ can be shown to be polynomial (if the functions $\pmb{r}_a$ can be computed in polynomial time for all arcs $a \in A$).
Indeed, if we consider a graph $(V,A)$ with arc costs $\pmb{r}_a$ for all $a \in A$ and capacities $c_a \in \mathbb{R}^+$, and a single commodity of demand $b \in \mathbb{R}^+$ with source node $s \in V$ and destination node $t \in V$, finding an optimal solution for the Unsplittable Convex Single Commodity Flow problem is equivalent to solving a shortest-path problem on a graph $(V,A')$, with $A'$ the subset of arcs $A$ with capacity $c_a$ greater than $b$, weighted by the cost $\pmb{r}_a(b)-\pmb{r}_a(0)$. As the shortest-path problem is polynomial, so is the Unsplittable Convex Single Commodity Flow problem.

\subsection{The Splittable-CMCF as a relaxation of the Unsplittable-CMCF}

Note that the  $\mathcal{INNER}$ formulation of the Splittable-CMCF problem can be converted to an Unsplittable-CMCF problem by extending the domain of $y_k$ and $x_k^p$ from $\{0,1\}$ to $\mathbb R^+$ for all $k \in K$ and $p \in P^k$, as in the formulation $\mathcal{CONVEX}$.
However, while the Splittable-CMCF problem is a relaxation of the Unsplittable-CMCF problem, it proves in practice to be a very weak one. This weakness is due to the convex nature of the arc cost functions. For instance, for quadratic cost functions, the cost of a path with very low traffic will be very low, which leads to a high dispersion of the flow on the arcs and thus the paths. This effect leads to a large number of fractional values of the $x_k^p$ variables in the optimal solution, as illustrated in Figure \ref{fig:path_var_hist} on the instance Germany50 of the SNDLib. When a linear cost function is optimized, $98\%$ of the path variables are in $\{0,1\}$ while, when quadratic cost functions are optimized, this share falls to $90\%$ of the total.

We will see in Section \ref{sect:tight_inner_apprx} and \ref{sect:patterns_problem} different inequalities or reformulations that allow tightening the problem and generate less fractional solutions.

\begin{figure}[ht]
\centering
\includegraphics[width=.6\textwidth]{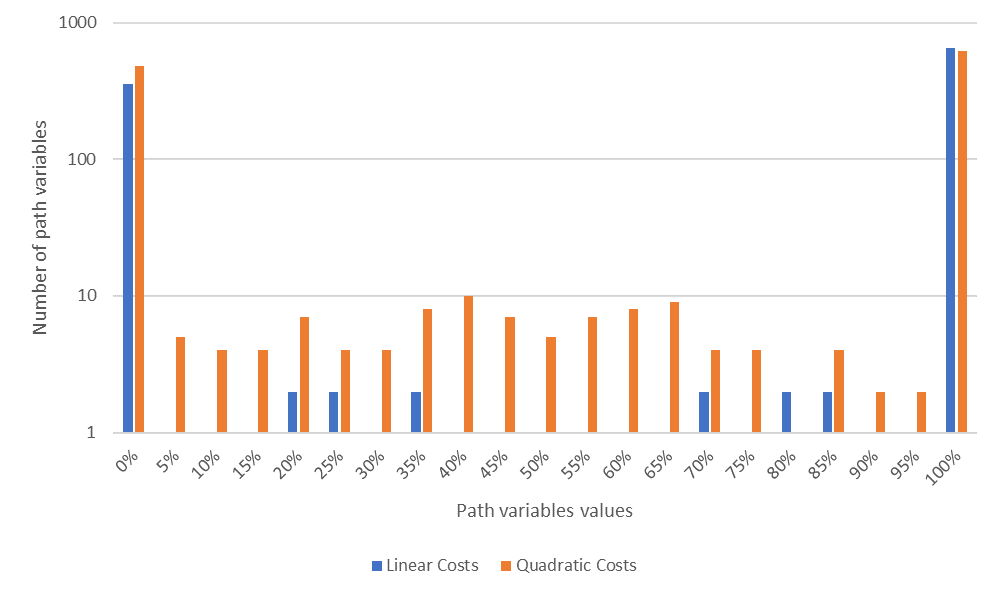}
    \caption{Histogram of the $x_k^p$ path variables values, for an optimal solution of the Splittable-CMCF problem on instance Germany50 \add{(|A|=176, |K|=662)}} 
    \label{fig:path_var_hist}
\end{figure}

\FloatBarrier

\subsection{Branching Strategy}\label{sect:branching}

In order to solve the Unsplittable-CMCF problem, one can embed the Splittable-CMCF problem described in the previous sections into a Branch-and-Price method to compute optimal integer solutions.
We will consider a branching tree on the extended formulations of the problem (either $\mathcal{CONVEX}$ or $\mathcal{INNER}$ formulations).
While the $\mathcal{CONVEX}$ or $\mathcal{INNER}$ formulations rely on the path variables, the branching will be applied to the underlying variables of the $\mathcal{COMPACT}$ formulation: a branch is defined by the possibility, for a given commodity $k \in K$, to be fully accepted or not in the network ($y_k=0$ or $y_k=1$) and another branching defines if the commodity crosses an arc $a \in A$ or not ($x_k^a > 0$ or $x_k^a=0$). The branching on the commodity acceptance is applied first and, if for all $k \in K$, $y_k \in \{0,1\}$ the branching on the \add{$x_k^a$\remove{$x_k^p$ and $y_k$}} variables is applied.

For a given node of the branching tree, the optimal solution of the relaxed problem is considered. If a commodity $k \in K$, \remove{with}\add{has} at least one associated variable from $\{{y_k}\} \cup \{{x_k^p}\}_{p \in P^k} $ in $]0,1[$, then the following branching rules can be applied. Two branchings strategies are considered, depending on the values of the optimal solution of the relaxed problem. In order to minimize the number of symetries generated by the branching tree, both branchings create partitions of the solutions of the Unsplittable-CMCF problem, as shown in Proposition \ref{prop:partition}.

\paragraph{Branching on the $y_k$ variables}
If there exists a single $p \in P^k$ such as the ${x_k^p}$ variable is in $]0,1[$, then two branches are created, one with $y_k=0$ and one with $y_k=1$.

\paragraph{Branching on the $x_k^p$ and $y_k$ variables}

If there are at least two paths $p$ and $p' \in P^k$ such as the $x_k^p$ and $x_k^{p'}$  variables are in $]0,1[$, then several branches are created, based on the paths associated with the non-null path variables. These paths are denoted $P_+^k = \{p | x_k^p \in ]0,1[ \}$. For each such commodity $k$, the divergence node $v_k$ (last node common to all paths in $P_+^k$ ) is defined similarly to \cite{barnhart2000}. The divergence node can be built by starting at the source node $s_k$ and checking if all paths of $P_+^k$ follow the same arc from $s_k$ to a node $u$. If they do, we iterate the same test from $u$ and the following nodes until we find a node from which at least two paths of $P_+^k$ follow different outgoing arcs. This vertex is the divergence node $v_k$.

The subpath between the source $s_k$ and the divergence node $v_k$ shared by all the paths of $P_+^k$ is called common path for the commodity $k$. It is denoted $p^{k+}$ below.

\remove{Several nodes are created}\add{From this relaxed solution, the branching rule creates the following nodes}:
\begin{enumerate}[label=(\Alph*),topsep=3pt]
    \item a new node of the branch-and-price tree is created by forcing commodity $k$ to follow the common path (i.e., $\forall a \in p^{k+}, a' \in \delta^+(tail(a)) \backslash p^{k+} : x_{a'}^k=0$), and forcing to use the arc \add{with origine $v_k$} $a \add{\in \delta^+(v_k)}$ with the highest usage ($ a =argmax_{a' \in \delta^+(v_k)}$ $ \sum_{p \in P^k | a' \in p} x_k^p $; if several arcs have the same load, one of them can be selected randomly) for commodity $k$ in the current RMP's optimal solution (i.e., $\forall a' \in \delta^+(v_k) \backslash \{a\}: x_{a'}^k=0$). The commodity $k$ is forced to be fully accepted by setting $y_k=0$. \label{branching:force_highest_arc}
    \item another new node of the branch-and-price tree is created by forcing commodity $k$ to follow the common path (this can be done as above) and preventing commodity $k$ from using the arc $a$ with origin $v_k$ and the highest usage (defined as above),  by setting $x_a^k=0$. The commodity is forced to be fully accepted by setting $y_k=0$. \label{branching:remove_highest_arc}
    \item another new node of the branch-and-price tree is created for every node $u$ in $p^{k+}$ ($u\in V | p^{k+} \cap \delta^+(u) \neq \emptyset$) by preventing commodity $k$ to follow the common path beyond $u$ ; this is achieved by setting $x_k^a=0$, with $a$ the arc in the common path with origin $u$. The commodity is forced to be fully accepted by setting $y_k=0$. \label{branching:cut_common_path}
    \item if the commodity $k$ is not forced to be fully accepted in the parent node, an additional node where the commodity is fully not accepted by setting $y_k=1$ is created.\label{branching:infeasible}
\end{enumerate}

\bigskip

These branching rules are illustrated in Figure \ref{fig:branchings}. For the sake of clarity, in the following, paths are described by the vertices it crosses separated by arrows. In Figure \ref{fig:branchings:paths}, 2 paths with non-null path variables for a given commodity, $s \shortrightarrow u \shortrightarrow v \shortrightarrow t$ and $s \shortrightarrow u \shortrightarrow v \shortrightarrow w \shortrightarrow t$, are shown in red ; path $s \shortrightarrow u \shortrightarrow v \shortrightarrow t$ carries more traffic for this commodity. The common path $s \shortrightarrow u \shortrightarrow v$ is shown with full lines. Figures \ref{fig:branchings:ruleA}, \ref{fig:branchings:ruleB} and \ref{fig:branchings:ruleC} illustrate the 3 branching rules \ref{branching:force_highest_arc}, \ref{branching:remove_highest_arc} and \ref{branching:cut_common_path} respectively: arcs removed from the graph for the considered commodity are shown in light dashed lines. Branching rule \ref{branching:infeasible} is not shown as the commodity would not enter this network in this branch.

Note that if the branching rule was only impacting arcs outgoing from the divergence node, as, for instance, in \cite{barnhart2000}, the path $s \shortrightarrow  w \shortrightarrow  t$ could appear in both branches.

\begin{figure}[ht]
    \centering
    \begin{subfigure}{.45\textwidth}
   \includegraphics[page=11,trim={0 7cm 0 6cm},clip, width=1\linewidth]{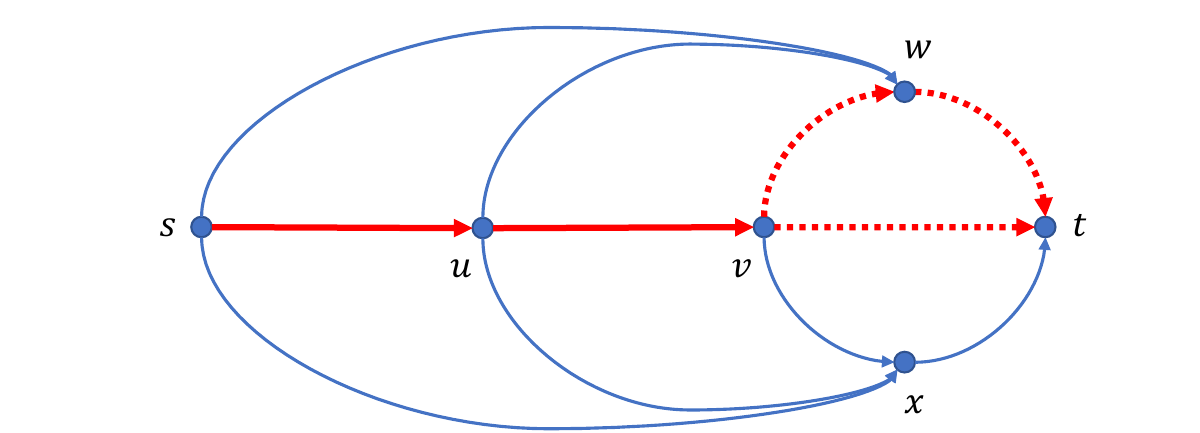}
   \caption{2 paths with non-null path-variables. The common path is in full red line, $v$ is the divergence node.} \label{fig:branchings:paths}
    \end{subfigure}
    \begin{subfigure}{.45\textwidth}
   \includegraphics[page=12,trim={0 7cm 0 6cm},clip, width=1\linewidth]{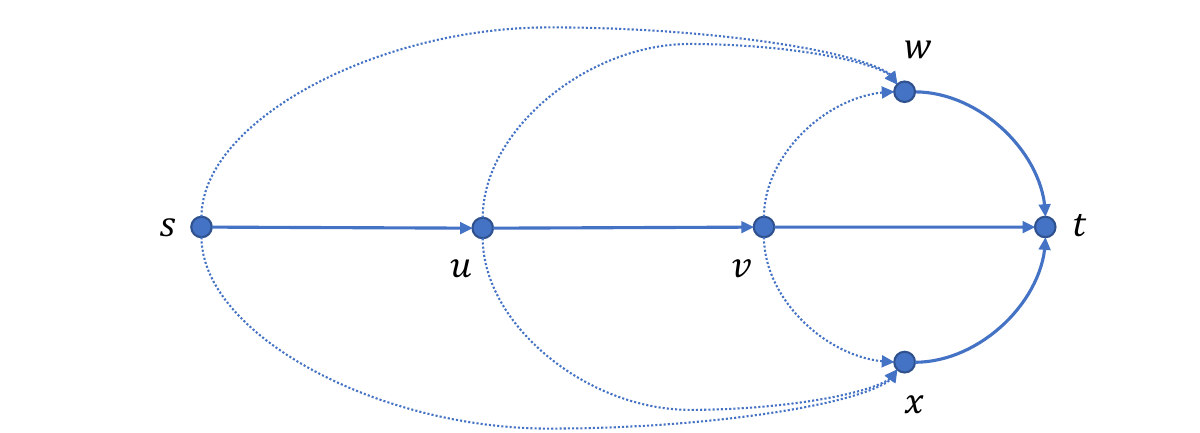}
   \caption{Branching rule \ref{branching:force_highest_arc} is applied \newline} \label{fig:branchings:ruleA}
    \end{subfigure}
    \begin{subfigure}{.45\textwidth}
   \includegraphics[page=13,trim={0 3.4cm 0 1.5cm},clip, width=1\linewidth]{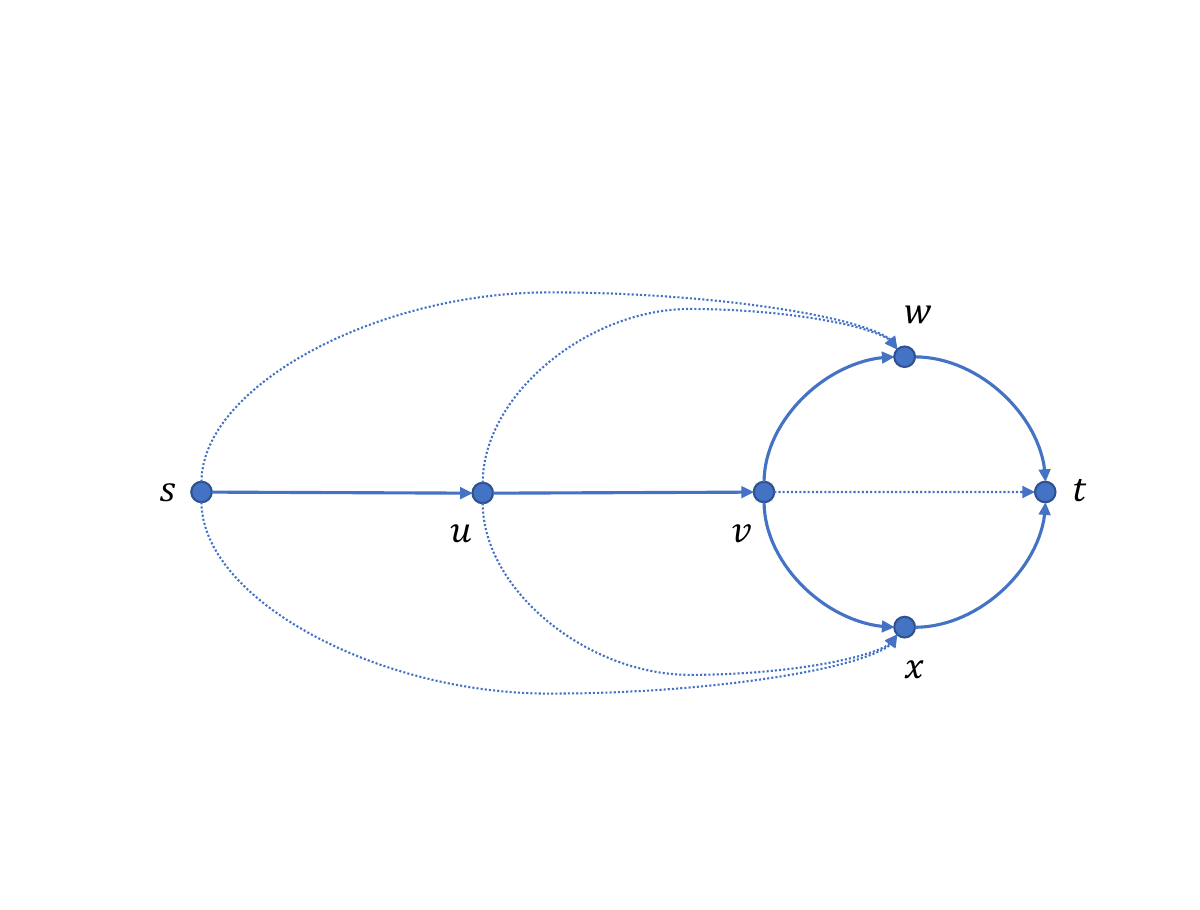}
   \caption{Branching rule \ref{branching:remove_highest_arc} is applied \newline} \label{fig:branchings:ruleB}
        \end{subfigure}
\begin{subfigure}{.45\textwidth}
   \includegraphics[page=14,trim={0 3.4cm 0 1.5cm},clip, width=1\linewidth]{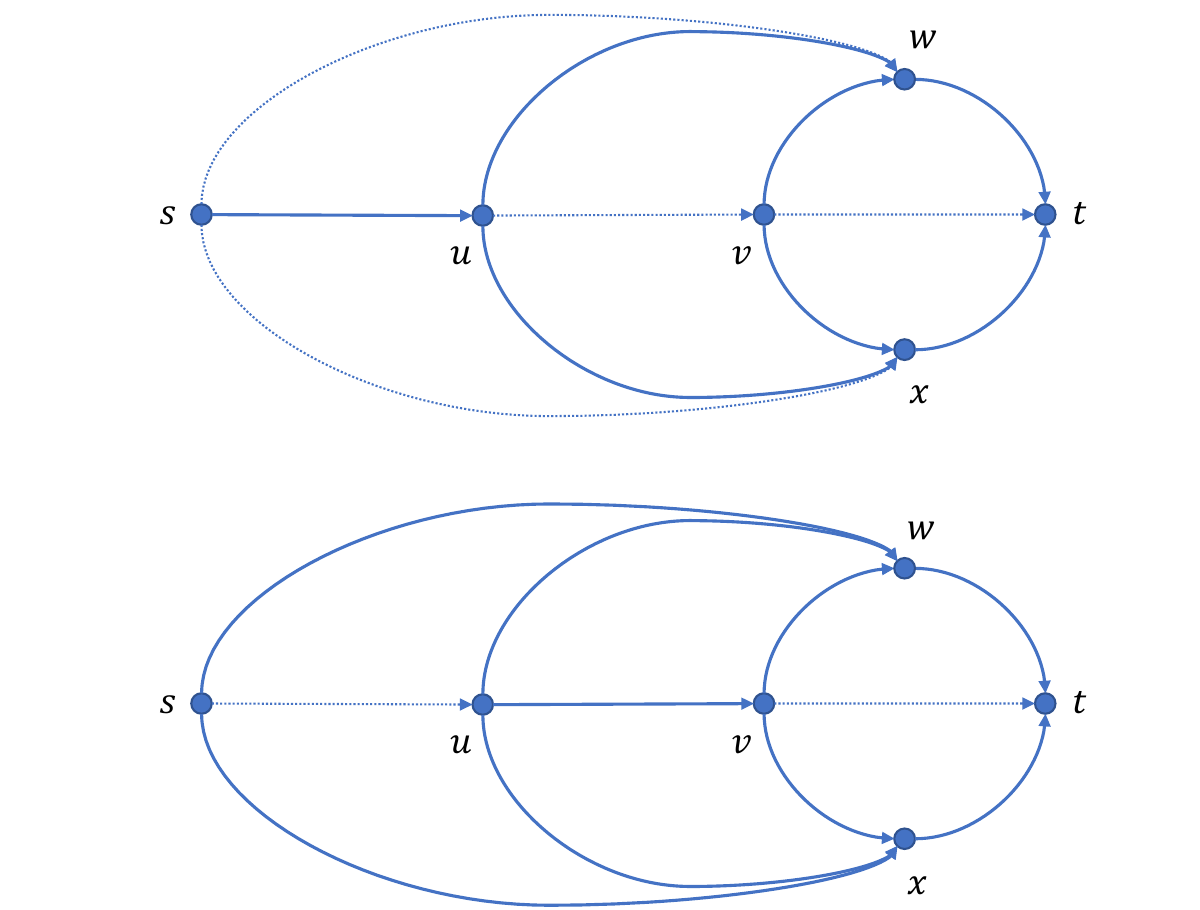}
 \caption{Branching rule \ref{branching:cut_common_path} is applied for nodes $u$ (top) and $s$ (bottom)} \label{fig:branchings:ruleC}
 \end{subfigure}
    \caption{Illustration of the branching rules on a toy instances. Links removed for the considered commodity are shown with light dashed lines.} 
    \label{fig:branchings}
\end{figure}

\FloatBarrier

\begin{proposition}
 Both branching strategies provide a partition of feasible solutions of the unsplittable-CMCF problem. \label{prop:partition}
\end{proposition} 

\begin{proof}
Consider a node of the tree branching on commodity $k \in K$, and an optimal solution of the extended formulation of the Unsplittable-CMCF problem for variables related to commodity $k$: $\{{y_k}^*\} \cup \{{x_k^p}^*\}_{p \in P^k}  \in \{0,1\}^{|P^k|+1}$.
Suppose that there exists a single path $p \in P^k$ such that ${x_k^p}^*$ is fractional. A branching on $y_k$ is applied, and two branches are created, one with $y_k=0$ and one with $y_k=1$. The two branches generated by a branching on a $y_k$ variable provide a partition of feasible solutions as $y_k \in \{0,1\}$ by the definition of problem $\mathcal{U-CONVEX}$.
If several $x_k^p$ variables are in $]0,1[$, a branching on $x_k^p$ and $y_k$ is applied. From Remark of Section \ref{rem:unique_var_at_one}, in all the optimal solutions of the Unsplittable-CMCF problem, one and only one of the variables in $\{{y_k}^*\} \cup \{{x_k^p}^*\}_{p \in P^k}$ equals one, all others equal zero. Lets show that each such solution is feasible in one and only one of the branches generated by a branching on the $x_k^p$ and $y_k$ variables.

Several cases can be considered: 
\begin{itemize}[topsep=3pt]
\item Optimal solutions with $y_k=1$ and $x_k^p = 0 \quad  \forall p \in P^k$ are included in the branch \ref{branching:infeasible} - this branch is either created at the considered node or higher in the branch-and-price tree. These solutions are not feasible in other branches.
\end{itemize}
In all other potentially optimal solutions, $y_k=0$ and, for one path denoted $\tilde{p} \in P^k$, variable $x_k^{\tilde{p}}=1$; for all other paths $p \in P^k \backslash \{\tilde{p}\}$, we have $x_k^p=0$.
The relaxed optimal solution for the considered node has a common path described above and written $p^{k+}$, between the commodity source $s_k$ and a divergence node $v_k$.
\begin{itemize}[resume*,topsep=3pt]
\item If the common path $p^{k+}$ is not a sub-path of $\tilde{p}$, there exists a node $u$ such that $\delta^+(u) \cap P^{k+} = \{a\}$ and $\delta^+(u) \cap P^{k+} = \{a'\}$ with $a \neq a'$. This solution is feasible only in branch \ref{branching:cut_common_path}.
\item If the common path $p^{k+}$ is a sub-path of $\tilde{p}$ and the arc in $\delta^+(v_k)$ with the highest usage is in $p^{k+}$, then the solution is feasible only in branch \ref{branching:force_highest_arc},
\item If the common path $p^{k+}$ is a sub-path of $\tilde{p}$ and the arc in $\delta^+(v_k)$ with the highest usage is not in $p^{k+}$, then the solution is feasible only in branch \ref{branching:remove_highest_arc}.
\end{itemize}

So, both proposed branchings offer partitions of potentially optimal solutions of the problem $\mathcal{U-CONVEX}$.
\unskip\nobreak\hfill $\square$
\end{proof}

\begin{corollary}
The branching described above is complete.
\end{corollary}
Although the second branching creates a relatively large number of nodes, it has the advantage over simpler branching strategies, such as the one described in \cite{barnhart2000}, of offering a partition of the potentially optimal solutions while cutting the fractional solution of the node considered. This avoids creating similar paths at different levels of the branching tree.

\subsection{Valid Inequalities for $\mathcal{U-INNER}$} \label{sect:tight_inner_apprx}

The $\mathcal{INNER}$ relaxation of the Unsplittable-CMCF problem can be tightened by leveraging the unsplittability hypothesis.
First, the paths pricing problem can be improved by requiring all paths generated to be unsplittable: all arcs on newly generated paths should be able to accept the demand's bandwidth.
If commodities are already forced to go through an arc $a \in A$ (i.e. $\forall p \in P^k: a \in p$), new paths for other demands should have sufficient residual capacity to cross $a$.

It is also possible to leverage the unsplittability constraints to further tighten the $\mathcal{INNER}$ relaxation of the Unsplittable-CMCF problem.  
As discussed in Section \ref{sect:inner_approx}, for every arc $a \in A$ at the optimum, a single vertex $i \in N$ per arc has a non-null coefficient $z_a^i$. This vertex is necessarily larger than the size $b_k$ of any \add{commodity \remove{demands}} $k \in K$ crossing the arc (i.e. $\forall a \in A, i \in N, b_k < c_a^i \implies z_a^i = 0$). If there is a demand $k \in K$ with $b_k > c_a^i$ and $z_a^i > 0$, then this demand would only be partially accepted.
This consideration allows us to replace the constraint

\begin{equation*}
\sum_{k \in K} \sum_{ p \in P^k | a \in p} b_k x_k^p - \sum_{i \in N} z_a^i c_a^i \leq 0 \qquad \forall a \in A
\end{equation*}

 by a set of constraints:

\begin{equation*}
\sum_{k \in K | b_k \geq f_j} \sum_{ p \in P^k | a \in p} b_k x_k^p  - \sum_{i \in N | c_a^i \geq f_j} z_a^i c_a^i \leq 0 \qquad \forall a \in A, j \in F
\end{equation*}

for any set of threshold values $\{{f_j}\}_{j \in F}$ with $F \subset \mathbb{N}^+$. 

In the tests presented in Section \ref{sect:perf_unsplit}, for example, a list of unique demand sizes $b_k$ is used for $F$ (as including other values in $F$ would only create redundant constraints). 
This problem, named $\mathcal{TIGHT-INNER}$, is by definition tighter than the formulation $\mathcal{INNER}$ presented above, while having the same unsplittable solutions. It does not change the polynomial nature of the formulation.

Including these constraints, the Restricted Master Problem becomes : 
\begin{align}
\mathcal{TIGHT-INNER}:  &    \min \sum_{a \in A} \sum_{i \in N} z_a^i r_a(c_a^i) + \sum_{k \in K} M b_k y_k \label{eq:tight_inner_obj}\\
& \eqref{eq:inner_demand_convexity}, \quad \eqref{eq:inner_polytop_convexity}, \quad \eqref{eq:inner_x_y_def}, \quad \eqref{eq:inner_z_def}, \nonumber \\
\beta_a^j : \quad &  \sum_{k \in K | b_k \geq f_j} \sum_{ p \in P^k | a \in p} b_k x_k^p - \sum_{i \in N | c_a^i \geq f_j} z_a^i c_a^i \leq 0 \qquad &\forall a \in A,j \in F,
\label{eq:tight_inner_usage_arc} 
\end{align}
where the Objective and all constraints except \eqref{eq:tight_inner_usage_arc} are similar to the $\mathcal{INNER}$ problem, and Constraints \eqref{eq:inner_usage_arc} are replaced by stronger Constraints \eqref{eq:tight_inner_usage_arc}, taking advantage of the unsplitability of flows. Note that Constraints \eqref{eq:inner_usage_arc} are included in Constraints \eqref{eq:tight_inner_usage_arc} for all $j \in F$ with $f_j=0$.
The dual variables $\beta$ are attached to these new constraints (with $\beta_a^j$ being the dual variable of the constraint for the arc $a$ and the threshold value $f_j$).
This change significantly increases the number of constraints in the Restricted Master Problem, but it  tightens the relaxation Unsplittable-CMCF problem. The balance between the weight of the problem and the tightening can be managed by the number of $f_j$ terms considered.

The pricing problems used to generate new paths and vertices can be derived from the above problem - and are very similar to the ones presented in \ref{sect:inner_approx}, replacing, for all arcs $a$, the $\beta_a$ terms by a sum of $\beta_a^j$. Using the same notation as above, where optimal dual values are marked with a $^*$, the pricing problems solved to generate new variables associated with paths or vertices are:
\begin{align*}
    \text{Find new path $p \in P^k$ such as: } \quad & \sum_{a \in p} \sum_{j \in F | b_k \geq f_j } {\beta_a^j}^* < \sfrac{\alpha_k^*}{b_k} \\
    \text{Find $c_a^i \in [0,c_a]$ such as: } \quad&   \pmb{r}_a(c_a^i) < \sum_{j \in F | c_a^i \geq f_j}{\beta_a^j}^* c_a^i - \gamma_a^*
\end{align*}
These two pricing problems are very similar to the ones described in the $\mathcal{INNER}$ column generation approach: the first one is a shortest-path problem, and the second one can be solved by a series of one-dimensional optimization problems similar to the ones described in Section \ref{sect:inner_approx}.

\subsection{Further tightening of the problem with demand patterns} \label{sect:patterns_problem}

As shown in Table \ref{tab:gap_and_perf}, the approaches described above lead to relaxations of the problem that are still too weak and prove inefficient when running a branch-and-price tree to solve the Unsplittable-CMCF problem.
A much tighter problem can be built by taking into account the total bandwidths of all commodities crossing an arc, and not only (as in the formulation $\mathcal{TIGHT-INNER}$) the bandwidth of each commodity taken independently. This formulation was first applied to Unsplittable linear Multi-Commodity Flow problems, for instance in \cite{park2003} or \cite{fan2024} but can be applied with even greater efficiency to the Unsplittable-CMCF problem.

To achieve this tighter relaxation, consider patterns of commodities $s \in S$ where $S$ is the set of all \add{subsets of $K$ - so each pattern $s \in S$ is a set of commodities \remove{sets of commodities}. The $\mathcal{U-PATTERN}$ formulation below will aim at finding the pattern, or set of commodities, associated to all arcs of the graph $G$}. For every arc $a \in A$ and pattern $s \in S$, we introduce a pattern variable $u_a^s \in \mathbb R^+$ representing the usage of pattern $s$ \remove{of commodities} over arc $a$. The problem optimizes the usage of different patterns, and the objective function is the sum over all arcs of the cost of all commodities present in the patterns (weighted by the pattern usage variable $u_a^s$). 
The Unsplittable-CMCF problem can be formulated as\remove{the following problem} : 
\begin{align}
\mathcal{U-PATTERN}:\quad  &    \min \sum_{a \in A} \sum_{s \in S}  \pmb{r}_a(\sum_{k \in s} b_k) u_a^s + \sum_{k \in K}  M b_k y_k, \label{eq:pattern_objective}\\
  & -\sum_{p  \in P^k} x_k^p - y_k \leq -1 \qquad & \forall k \in K, \label{eq:pattern_commodities_convexity}\\
      &  \sum_{p \in {P^k} | a \in p} x_k^p - \sum_{s \in S | k \in s} u_a^s \leq 0 \qquad & \forall a \in A, k \in K, \label{eq:pattern_commo_pattern_arc}\\
  & \sum_{s \in S} u_a^s = 1 \qquad &\forall a \in A, \label{eq:pattern_pattern_convexity}\\
    & x_k^p \in \{0,1\} \qquad & \forall k \in K,p \in P^k,\\
    & y_k \in \{0,1\} \qquad & \forall k \in K,\\
    & u_a^s \geq 0 \qquad & \forall a \in A, s \in S. \label{eq:pattern_u_def}
\end{align}
The Objective \eqref{eq:pattern_objective} sums, for all arcs, the cost of using patterns $\pmb{r}_a(\sum_{k \in s} b_k)$ weighted by the patterns usage variables $u_a^s$. As in the previous formulation, Constraints \eqref{eq:pattern_commodities_convexity} ensure that each commodity is accepted (or that the corresponding variables $y_k$ are non-null). Constraints \eqref{eq:pattern_commo_pattern_arc} limit the use of a path $p \in P^k$ to the level of use of the patterns containing commodity $k$ for all arcs of the path, while Constraints \eqref{eq:pattern_pattern_convexity} ensure that, for each arc, the sum of the patterns usages is one.

As shown in Proposition \ref{prop:pattern_01} below, the patterns variables can be defined in $\mathbb{R}^+$, and a branch-and-price tree can be built only on the $x_k^p$ and $y_k$ variables for $k \in K$ and $p \in P^k$ to solve the Unsplittable-CMCF problem.
\begin{proposition}
    If, for an optimal solution of $\mathcal{PATTERN}$, for all commodities $k \in K$ and paths $p \in P^k$, $x_k^p \in \{0,1\}$ and $y_k \in \{0,1\}$, then for all arcs $a \in A$ and patterns $s \in S$,  $u_a^s \in \{0,1\}$. \label{prop:pattern_01}
\end{proposition}
\begin{proof}
Let us consider the contrary. There exists an optimal solution of the problem $\mathcal{PATTERN}$ with  $y_k \in \{0,1\}$ and $x_k^p \in \{0,1\}$ for all $k \in K$ and $p \in p^k$, and a variable $u_a^s \in ]0,1[$ for some arc $a$ and the pattern of commodities $s$.

The Constraints \eqref{eq:pattern_pattern_convexity} and \eqref{eq:pattern_u_def} ensure that at least another variable $u_a^{s'}$, with pattern $s' \neq s$ is in $]0,1[$.

The Constraints \eqref{eq:pattern_commo_pattern_arc} ensure that all commodities $k \in K$ with $\sum_{p \in {P^k} | a \in p} x_k^p = 1$ are in both $s$ and $s'$.

Let us call $\hat{s}$ the pattern that contains all commodities $k$ with $\sum_{p \in {P^k} | a \in p} x_k^p = 1$ and no other.
We know that $\hat{s} \subset{s}$ and $\hat{s} \subset{s'}$, so $s$ or $s'$ contains a commodity not in $\hat{s}$. As the cost function $\pmb{r}_a$ increases and $\sum_{k \in s} b_k > \sum_{k \in \hat{s}} b_k$ or $\sum_{k \in s'} b_k > \sum_{k \in \hat{s}} b_k$, setting $u_a^{\hat{s}}=1$ and both $u_a^s=0$ and $u_a^{s'}=0$ (if $s$ and $s' \neq \hat{s}$), we build a solution that is feasible and dominates the one previously considered.
So all optimal solutions of $\mathcal{PATTERN}$ with, for all commodities $k \in K$ and paths $p \in P^k$, $x_k^p \in \{0\add{,\remove{.}}1\}$ and $y_k \in \{0,1\}$, should have $u_a^s \in \{0,1\}$ for all arcs $a$ and pattern $s$.
\unskip\nobreak\hfill $\square$
\end{proof}

The formulation $\mathcal{U-PATTERN}$ can be relaxed as $\mathcal{PATTERN}$ by relaxing the domain of the $x_k^p$ and $y_k$ variables from $\{0,1\}$ to $[0,1]$ for all commodities $k \in K$ and paths $p \in P$:
\begin{align*}
 \mathcal{PATTERN}: \quad  &    \min \sum_{a \in A} \sum_{s \in S}  \pmb{r}_a(\sum_{k \in s} b_k) u_a^s + \sum_{k \in K}  M b_k y_k, \\
&\eqref{eq:pattern_commodities_convexity}, \quad  \eqref{eq:pattern_commo_pattern_arc}, \quad  \eqref{eq:pattern_pattern_convexity} \\
& x_k^a \geq 0  \quad \forall a \in A, k \in K, \qquad y_k \geq 0 \quad   \forall  k \in K,  \qquad u_a^s \geq 0 \quad  \forall a \in A, s \in S. 
\end{align*}
This representation requires identifying the patterns used for all arcs ; it also requires $|K| \times |A|$ constraints to manage the relations between the paths and pattern usage variables.
As both pattern and path variables grow exponentially with the size of the problem, they need to be generated via a column-generation scheme.

The pattern formulation presented above relies on solving two pricing problems, a shortest path problem to generate new path variables, and a convex knapsack problem to generate new pattern variables.
Denoting $\alpha_k$, $\beta_a^k$ and $\gamma_a$ the dual variables associated with Constraints \eqref{eq:pattern_commodities_convexity}, \eqref{eq:pattern_commo_pattern_arc} and \eqref{eq:pattern_pattern_convexity}, the pricing problems can be written as:
\begin{align*}
    \text{Find a path $p \in P^k$ such as:}   \qquad& \sum_{a \in p} b_k {\beta_a^k}^* < \alpha_k^* \\
    \text{Find a pattern $s \in S$ such as:} \qquad& \sum_{k \in s} {\beta_a^k}^* - \pmb{r}_a(\sum_{k \in s} b_k) > \gamma_a^*
\end{align*}
While the first problem is a shortest path problem that can be solved in polynomial time, the second pricing problem is a the lagrangian of a non-linear knapsack problem, which is known, for example in \cite{hochbaum1995}, to be pseudopolynomial. This makes the pricing problems more computationally expensive than the $\mathcal{INNER}$ formulation, but still tractable.
However, the number of generated sets may increase exponentially due to symmetries between \add{commodities \remove{demands}}.

\FloatBarrier

\begin{figure}[ht]
\centering
\includegraphics[page=  15,trim={3cm 6.8cm 3cm 7.5cm},clip, width=.6\textwidth]{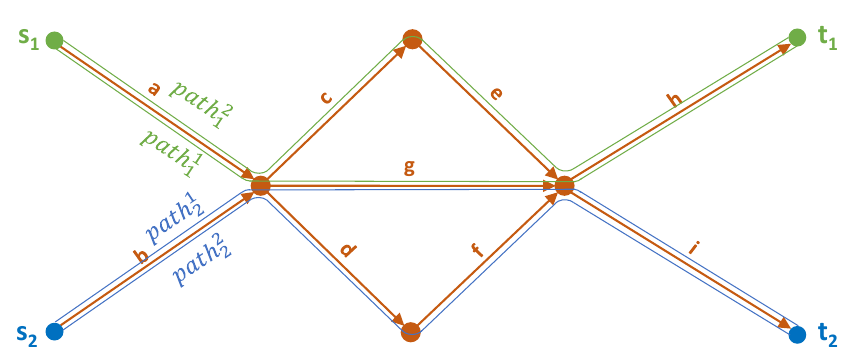}
    \caption{\add{Example of instance with two commodities of size $1$, with two paths each. All edges have a capacity greater than $2$.}} 
    \label{fig:two_paths}
\end{figure}

\add{Figure \ref{fig:two_paths} illustrates the different relaxations on a toy instance, with two commodities; commodity $1$ has a size $b_1=1$ and commodity $2$ has a size $b_2=2$. For each commodity, two paths are considered, one of which crosses the arc $g$. For the sake of simplicity, we will suppose that the optimal solutions of this example splits each commodity equally between its two paths (so at the optimum, the values of the path usage variables $x_k^p$ used in the extended formulations equal $x_1^1=x_1^2=0.5$ for commodity $1$ and $x_2^1=x_2^2=0.5$ for commodity $2$). We can compare the cost induced by arc $g$ through the associated cost function $\pmb{r}_g$ in the relaxations $\mathcal{CONVEX}$, $\mathcal{PATTERN}$, $\mathcal{TIGHT-INNER}$ and $\mathcal{INNER}$. \\
In $\mathcal{CONVEX}$, the cost associated with arc $g$ is given by $\pmb{r}_g(x_1^1 b_1 + x_2^1 b_2)=\pmb{r}_g(1.5)$.\\
In $\mathcal{INNER}$, the load of arc $g$ is represented as a convex representation of load values $c_g^i$ weighted by variables $z_g^i$. The cost associated with arc $g$ is formulated as $z_g^1 \pmb{r}_g(c_g^1)= 1 \times \pmb{r}_g(1.5)$  $\qquad$  (taking $c_g^1=1.5$, and $z_g^1=1$), so Constraint \eqref{eq:inner_usage_arc}, stating that $z_g^1 c_g^1 = 1.5 \geq x_1^1 b_1 + x_2^1 b_2 = 1.5$ holds true. Note that the optimal values of $\mathcal{CONVEX}$ and $\mathcal{INNER}$ are equal, as seen in Section \ref{sect:inner_approx}.\\
In $\mathcal{TIGHT-INNER}$, two arc load values $c_g^1=1$ and $c_g^2=2$ are considered, with their associated weighting variables $z_g^1$ and $z_g^2$. The cost of arc $g$ becomes $z_g^1 + \pmb{r}_g(c_g^1) + z_g^2 \pmb{r}_g(c_g^2)= 0.5 \pmb{r}(1) + 0.5 \pmb{r}(2)$ $\qquad$ ($z_g^1=0.5$ and $z_g^2=0.5$), so Constraints \eqref{eq:tight_inner_usage_arc} hold true: $z_g^1 c_g^1+z_g^1 c_g^1=0.5 \times 1+0.5 \times 2 = 1.5 \geq x_1^1 b_1 + x_1^2 b_2 = 1.5$, and $z_g^1 c_g^1=0.5 \times 2 \geq x_1^2 b_2 = 1$. Because the cost function $\pmb{r}_g$ is convex, this relaxation gives an higher cost than the previous ones.\\
Finally, in $\mathcal{PATTERN}$, the pattern containing both commodities, $s = \{1, 2\}$ is considered, and its associated coefficient $u_g^s$ equals 0.5 in the optimal solution defined above. So the cost associated with arc $g$ becomes: $u_g^s \pmb{r}_g(b_1+b_2)=0.5 \pmb{r}_g(3)$, and constraint \eqref{eq:pattern_commo_pattern_arc} holds true: $u_g^s = 0.5 \geq x_1^1 = 0.5$, and $u_g^s = 0.5 \geq x_2^1 = 0.5$. Once again, because of the convexity of the cost function $\pmb{r}_g$, this formulation leads to an higher cost than the previous ones.\\
Figure \ref{fig:unsplit_comparison} gives visual representations of the computation of the cost associated to arc $g$ with the different relaxations.}

\begin{figure}[ht]
\centering
\begin{subfigure}{.32\textwidth}
  \centering
  \captionsetup{justification=centering}

  \includegraphics[page=16,trim={0cm 4cm 0cm 4.5cm},clip, width=1\linewidth]{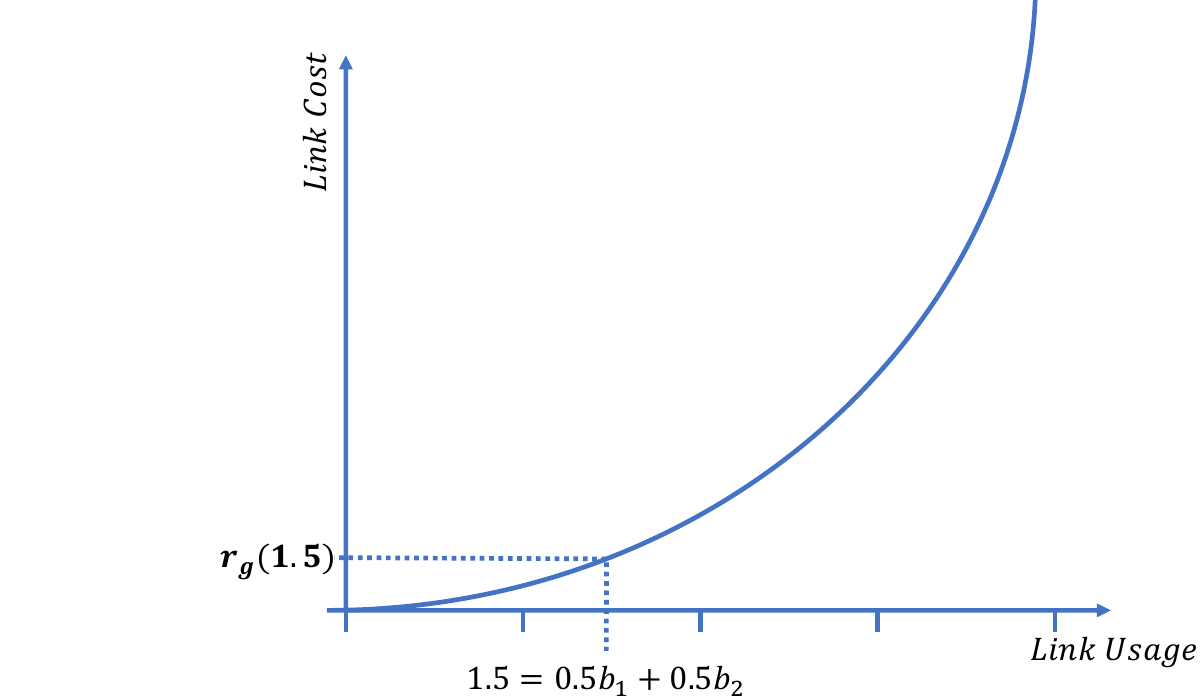}
  \caption{$\mathcal{CONVEX}$ and $\mathcal{INNER}$}
\end{subfigure}%
\begin{subfigure}{.32\textwidth}
  \centering
  \captionsetup{justification=centering}

  \includegraphics[page=17,trim={0.cm 4cm 0cm 4.5cm},clip, width=1\linewidth]{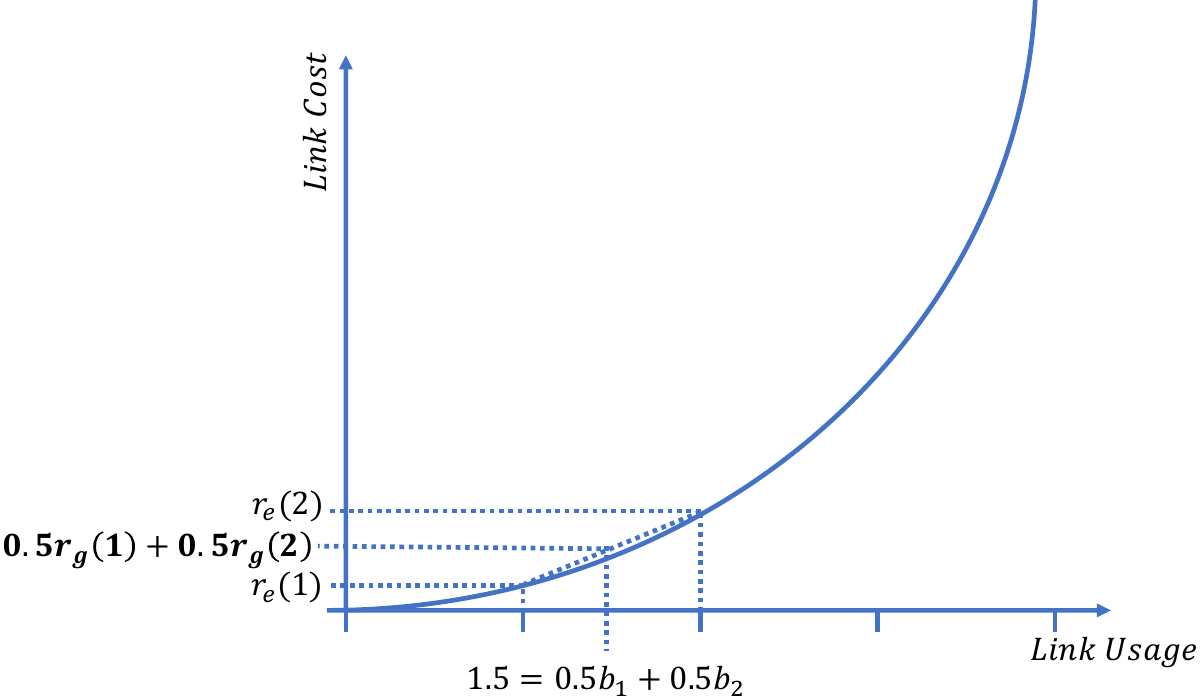}
  \caption{$\mathcal{TIGHT-INNER}$}
\end{subfigure}
\begin{subfigure}{.32\textwidth}
  \centering
  \captionsetup{justification=centering}

  \includegraphics[page=18,trim={0.cm 4cm 0cm 4.5cm},clip, width=1\linewidth]{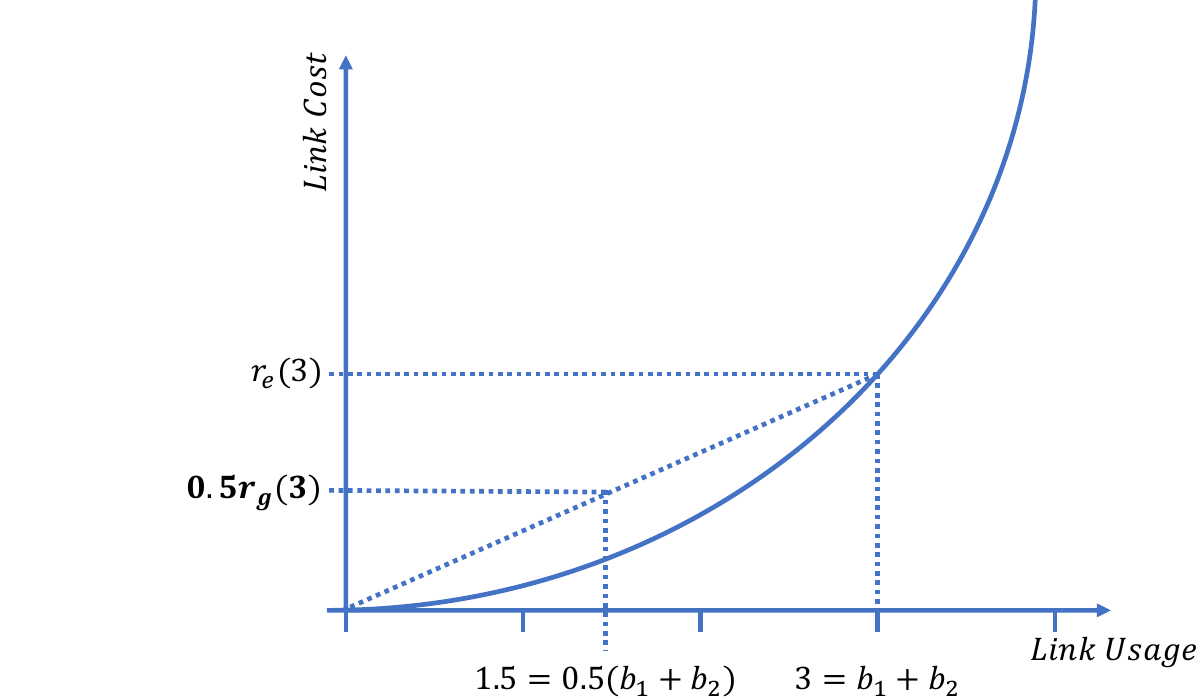}
  \caption{$\mathcal{PATTERN}$}
\end{subfigure}
\caption{\add{Illustration of the cost of arc $g$ in Figure \ref{fig:two_paths} for relaxations of the Unsplittable-CMCF problem.}}
\label{fig:unsplit_comparison}
\end{figure}

\FloatBarrier

\subsection{Heuristic upper bound} \label{sect:greedy_heuristic}

A heuristic solution can be computed with a greedy algorithm, where commodities are accepted sequentially in the network. 
When considering a given commodity $k \in K$, we note $x_a$ the total flow of all previously accepted commodities over an arc $a\in A$, and we define a cost $r_a$ as the marginal cost of commodity $k$ if it can cross arc $a$: $r_a = \pmb{r}_a(x_a+b_k) - \pmb{r}_a(x_a) \text{ if } x_a+b_k \leq c_a \text{ and } +\infty$ if $x_a+b_k > c_a $.
A shortest path $p$ from $s_k$ to $t_k$, weighted by $r_a$ for all arcs $a \in A$ is computed. If its total cost is finite, the commodity is accepted in the network, and its bandwidth $b_k$ is added to the $x_a$ for $a \in p$, and the next commodity is considered.

The algorithm \ref{algo:greedy} describes a run of the greedy algorithm, where the function \textsc{ShortestPath}$(s_k,t_k,V,A,r)$ returns the shortest paths between nodes $s_k$ and $t_k$ in graph $(V,A)$ weighted by the coefficients $\{r_a\}_{a \in A}$ (the cost of this path may be infinite). Function \textsc{Shuffle}$(K)$ return a list of randomly ordered demand in the set $K$. Random reordering ensures that different calls to the \textsc{GreedyHeuristic} function explore different initial conditions.

\begin{algorithm}
\small{

\begin{algorithmic}
  \State $K' \gets \texttt{Shuffle}(K)$
  \State $\forall k \in K': \quad y_k \gets 0, \qquad \forall a \in A: \quad x_a \gets 0$
      \For{$k \in K'$}
          \State $\forall a \in A: \quad   \algorithmicif \quad  x_a + b_k \leq c_a \quad  \algorithmicthen \quad r_a \gets \pmb{r}_a(x_a+b_a) - \pmb{r}_a(x_a) \quad \algorithmicelse \quad r_a \gets+ \infty
          $
         
        \State $p_k \gets \texttt{ShortestPath}(s_k, t_k, V, A, \{r_a\}_{a \in A})$
        \If {$\sum_{a \in p_k} r_a < + \infty$}
          \State $\forall a \in p_k: \quad x_a \gets x_a + b_k$
        \Else
          \State $y_k \gets 1$
        \EndIf
      \EndFor
    \State \Return $\sum_a \pmb{r}_a(x_a) + \sum_{k \in K}  M b_k y_k$

\end{algorithmic}    
    } 
  \caption{Greedy Algorithm for the Unsplittable-CMCF Problem. }
  \label{algo:greedy}
\end{algorithm}

This greedy heuristic result is highly dependent on the order of the commodities $K$, but there may not be any order of commodities that provides the optimal solution.
For example, in the instance described in Figure \ref{fig:instance_greedy}, all orders of the commodities lead to a suboptimal solution.
3 commodities with source $s$ and destination $t$ are considered, with bandwidths $b_1=3$, $b_2=2$ and $b_3=2$. Arc costs functions $\pmb{r}_a$ are described in the figure.
Table \ref{tab:greedy_heuristic} describes the results of Algorithm \ref{algo:greedy}: column "Commodity Order" describes the order of commodities in the set $K$ provided to the algorithm~; columns "Paths" describe the paths created for each commodity, while columns "Arcs Final Usage" and "Arcs Final Cost" describe the usages and costs of each arc in the solution generated. The last column gives the objective value returned by the objective.
As commodities 1 and 2 have the same size, the table describes all possible solutions generated by algorithm \ref{algo:greedy}.
A feasible solution with a lower objective value is also provided (line "Optimal solution").

\begin{figure}[ht]
    \begin{subfigure}{.4\textwidth}
        \includegraphics[page=5,trim={2.8cm 2.5cm 1cm 1.5cm},clip, width=1\linewidth]{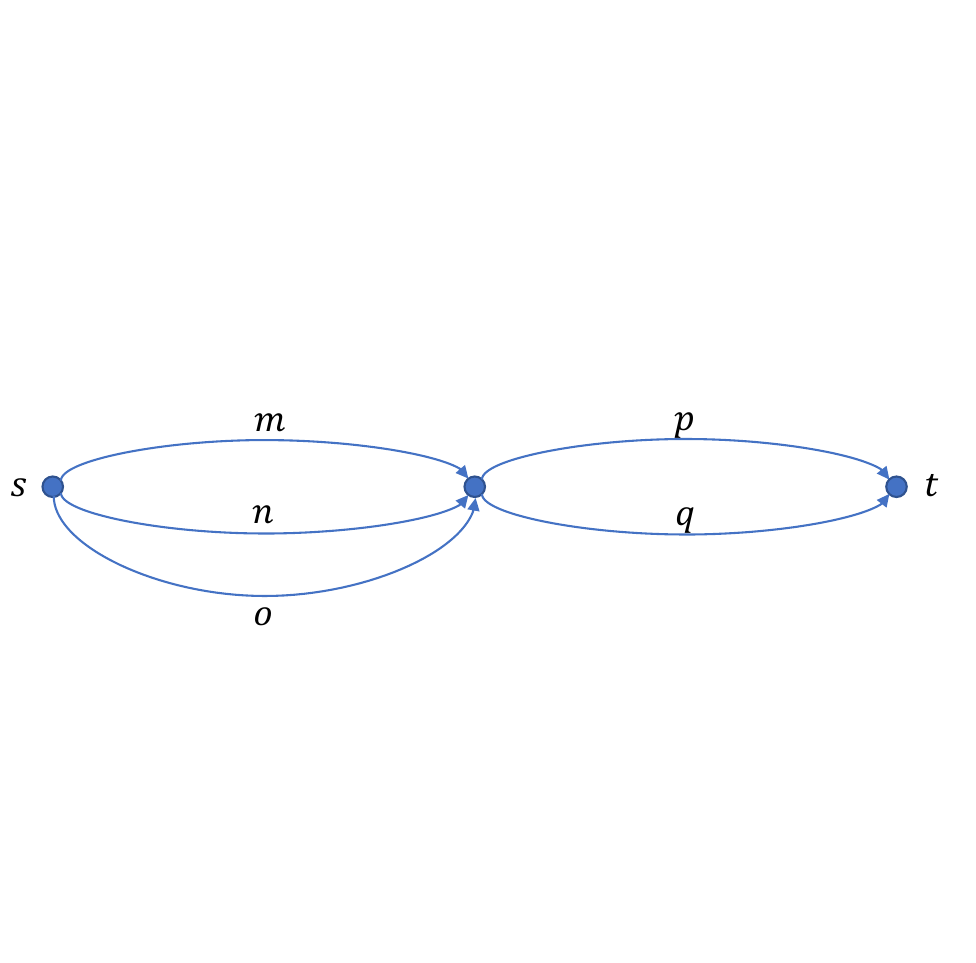}
        \caption{Graph of the example instance.}
        \label{fig:ex_inst_alone_L}
    \end{subfigure}
    \begin{subfigure}{.6\textwidth}
        \includegraphics[page=6,trim={1.5cm 5.cm 1cm 1.5cm},clip, width=1\linewidth]{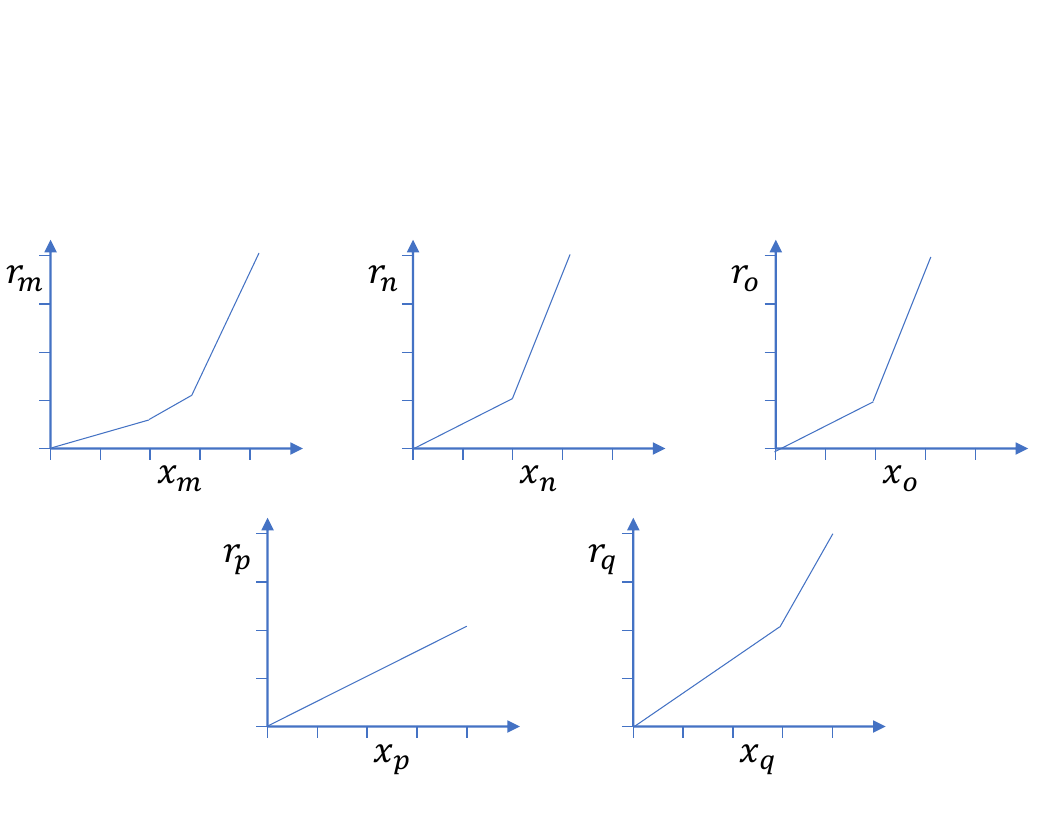}
        \caption{Arc costs for the example instance.}
        \label{fig:ex_inst_alone_R}
    \end{subfigure}
    \caption{Example of instance with 3 commodities of size 2, 2 and 3 from $s$ to $t$, where the heuristic defined by Algorithm \ref{algo:greedy} does not provide any optimal solution.}
    \label{fig:instance_greedy}
\end{figure}

\FloatBarrier

\begin{table}[H]
\small{
\begin{tabular}{c|ccc|ccccc|ccccc|c}
\multirow{2}{*}{\makecell{Commodities\\Order}}	&	\multicolumn{3}{c|}{Paths} & \multicolumn{5}{c|}{Arcs Final Usage} & \multicolumn{5}{c|}{Arcs Final Cost} & \multirow{2}{*}{\makecell{Objective \\ Value}} \\
           & 1 & 2 & 3	& m & n & o & p & q & m & n & o & p & q &       
           \\
\hline
$1-2-3$ & $m-p$ & $n-q$ & $o-q$ & 3 & 2 & 2 & 3 & 4 & 1 & 1 & 1 & 1.5 & 4 & 8.5\\
$2-1-3$ & $n-q$ & $m-p$ & $o-p$ & 2 & 3 & 2 & 4 & 3 & 0.5 & 2.5 & 1 & 2 & 2 & 8\\
$2-3-1$ & Infeasible & $m-p$ & $n-q$ & 2 & 2 & 0 & 2 & 2 & 0.5 & 1 & 2.5 & 2 & 2 & $>3M$ \\
 &  &  &  &  &  &  &  &  &  &  &  &  &  &  \\
Opt. solution &  $m-q$ & $n-p$ & $o-p$ & 3 & 2 & 2 & 4 & 3 & 1 & 1 & 1 & 2 & 2 & 7\\
\end{tabular}
} 

\caption{Results of Algorithm \ref{algo:greedy} and an optimal solution for the instance described Figure \ref{fig:instance_greedy}.}
\label{tab:greedy_heuristic}
\end{table}

\FloatBarrier

\subsection{Performances: Unsplittable Convex Multi-Commodity Flow}\label{sect:perf_unsplit}

The different relaxations of the Unsplittable-CMCF problem were implemented and tested on SNDLib instances, described in \ref{sect:perf_split_compact_convex}.  
Two different relaxations, $\mathcal{TIGHT-INNER}$ and $\mathcal{PATTERN}$, were considered to solve the Unsplittable-CMCF problem in a branch-and-price algorithm described in Section \ref{sect:branching}. 
Performances on the subset of instances of SNDLib with less than $400$ \add{commodities \remove{demands}}, described in Table \ref{tab:instances_descriptions}, are shown in Table \ref{tab:gap_and_perf}.
Directly solving the $\mathcal{U-COMPACT}$ problem with CPLEX proves inefficient for instances with more than a few links. 

To compare the tightness of the two formulations, the relative gap is defined as:
\begin{equation*}
    gap=\frac{R^*-S^*}{U^*-S^*}
\end{equation*}
where $S^*$, $U^*$ and $R^*$ are the optimal values of the splittable-CMCF problem, the unsplittable-CMCF problem and the relaxation considered, respectively.
Relative gaps are expected to be in $[0\%,100\%]$ for relaxations of the Unsplittable-CMCF problem.
If the relative gap of a relaxation is $0\%$, it does not provide any tightening compared to the Splittable problem, while if the relative gap is $100\%$, the relaxation's optimal value matches the Unsplittable-MCF optimal value. 
On the other hand, any feasible solution of the Unsplittable-CMCF problem, like the heuristic described in Section \ref{sect:greedy_heuristic}, have a relative gap of at least $100\%$ (and equal to $100\%$ if the heuristic provides the optimal solution to the problem).

Comparisons of performances in Table \ref{tab:gap_and_perf} illustrate the superiority of the pattern approach. First, at the root node level, the solutions of this relaxation always display higher relative gaps than the $\mathcal{TIGHT-INNER}$ ones. Second, at the branch-and-price tree level, the $\mathcal{PATTERN}$ relaxation allows to generate manageable branch-and-price trees on all instances of the tested instances but two.
On the contrary, the $\mathcal{TIGHT-INNER}$ relaxation either directly provides an unsplittable solution, in which case its relative gap is $100\%$, or leads to branch-and-price trees that do not reach the target precision of $0.1\%$ after $24$ hours. For instances that did not reach the $0.1\%$ precision after $24$ hours, the branch-and-price gap (defined as the ratio between the upper and lower bounds of the tree) is displayed between parenthesis.

For this reason, the performances of branch-and-bound trees embedding the $\mathcal{TIGHT-INNER}$ formulation are not shown.

However, the tightness of the $\mathcal{PATTERN}$ relaxation comes at the expense of its solving duration, which can be significantly higher than the $\mathcal{TIGHT-INNER}$ relaxation: on the instances tested, the former could be up to $2000$ time slower than the later.

\begin{table}
\centering
\scriptsize{
\begin{tabular}{|c|cccc|}
\hline
Instance Name & Nodes & Arcs & \add{Commodities} & Arcs $\times$  \add{Commodities} \\
\hline
pdh		&  11	&  34	&  24	&  816	\\
di-yuan		&  11	&  42	&  22	&  924	\\
polska		&  12	&  18	&  66	&  1188	\\
nobel-us		&  14	&  21	&  91	&  1911	\\
abilene		&  12	&  15	&  132	&  1980	\\
nobel-germany		&  17	&  26	&  121	&  3146	\\
dfn-bwin		&  10	&  45	&  90	&  4050	\\
atlanta		&  15	&  22	&  210	&  4620	\\
dfn-gwin		&  11	&  47	&  110	&  5170	\\
sun		&  27	&  102	&  67	&  6834	\\
newyork		&  16	&  49	&  240	&  11760	\\
ta1		&  24	&  55	&  396	&  13500	\\
\hline
\end{tabular}
}
\caption{Instances of the SNDLib with $|K|\leq 400$, tested for the Unsplittable-CMCF Problem}
\label{tab:instances_descriptions}
\end{table}

\begin{table}
\small{

\begin{subtable}{\linewidth}

\caption{Kleinrock cost functions} \label{subtab:kr_results}

\makebox[\textwidth][c]{
\scriptsize{
\begin{tabular}{|c|ccc|cc|cc|}
\hline
Instance &	 \multicolumn{3}{c|}{Root Node or Heuristic Relative Gaps} & \multicolumn{2}{c|}{Root Node Duration in minutes } & \multicolumn{2}{c|} {$\mathcal{PATTERN}$ Branch-and-Bound}\\
& $\makecell{\mathcal{TIGHT-}\\ \mathcal{INNER}}$ & $\mathcal{PATTERN}$	& \makecell{Greedy \\ Heuristic}  &  \makecell{$\mathcal{TIGHT-}$\\ $\mathcal{INNER}$} & $\mathcal{PATTERN}$ & \makecell{Duration in minutes \\ (B\&P Gap)} & \#Nodes \\
\hline
pdh							&  100\%	&  100\%	&  100\%		&  $<$ 0.1 	&  $<$ 0.1 	&  $<$ 0.1 	&  1	\\
di-yuan							&  100\%	&  100\%	&  100\%		&  $<$ 0.1 	&  $<$ 0.1 	&  $<$ 0.1 	&  1	\\
polska							&  1\%	&  91\%	&  458\%		&  0.2	&  0.7	&  2.3	&  38	\\
nobel-us							&  1\%	&  91\%	&  200\%		&  0.1	&  2.7	&  9.6	&  56	\\
abilene							&  44\%	&  99\%	&  143\%		&  0.5	&  95.3	&  148.1	&  5	\\
nobel-germany							&  52\%	&  99\%	&  168\%		&  $<$ 0.1 	&  9.8	&  38.9	&  82	\\
dfn-bwin							&  100\%	&  100\%	&  100\%		&  2.2	&  0.4	&  0.4	&  1	\\
atlanta							&  1\%	&  91\%	&  1863\%		&  0.8	&  241.1	&  645.6	&  10	\\
dfn-gwin							&  100\%	&  100\%	&  100\%		&  0.8	&  0.5	&  0.5	&  1	\\
sun							&  6\%	&  83\%	&  893\%		&  0.1	&  4.1	&  11.5	&  80	\\
newyork							&  57\%	&  $\geq$91.8\%	&  $\geq$1395230.9\%		&  0.2	&  113.8	&  (0.5\%)	&  2660	\\
ta1							&  100\%	&  100\%	&  100\%		&  2.7	&  0.8	&  0.8	&  1	\\
\hline
\end{tabular} 
} 
} 
\end{subtable}

\begin{subtable}{\linewidth}
\caption{Quadratic cost functions} \label{subtab:sq_results}
\makebox[\textwidth][c]{
\scriptsize{
\begin{tabular}{|l|ccc|cc|cc|}
\hline
Instance	&	 \multicolumn{3}{c|}{Root Node or Heuristic Relative Gaps} & \multicolumn{2}{c|}{Root Node Duration in minutes} & \multicolumn{2}{c|} {$\mathcal{PATTERN}$ Branch-and-Bound}\\
& $\makecell{\mathcal{TIGHT-}\\ \mathcal{INNER}}$ & $\mathcal{PATTERN}$	& \makecell{Greedy \\ Heuristic}  &  \makecell{$\mathcal{TIGHT-}$\\ $\mathcal{INNER}$} & $\mathcal{PATTERN}$ & \makecell{Duration in minutes \\ (B\&P Gap)} & \#Nodes \\
\hline		
pdh							&  100\%	&  100\%	&  100\%		&  $<$ 0.1 	&  $<$ 0.1 	&  $<$ 0.1 	&  1	\\
di-yuan							&  100\%	&  100\%	&  100\%		&  $<$ 0.1 	&  $<$ 0.1 	&  $<$ 0.1 	&  1	\\
polska							&  7\%	&  97\%	&  415\%		&  $<$ 0.1 	&  0.5	&  0.6	&  3	\\
nobel-us							&  8\%	&  100\%	&  204\%		&  $<$ 0.1 	&  2.7	&  2.7	&  1	\\
abilene							&  44\%	&  99\%	&  116\%		&  0.3	&  47	&  124.1	&  11	\\
nobel-germany							&  83\%	&  100\%	&  142\%		&  $<$ 0.1 	&  18.5	&  60.8	&  48	\\
dfn-bwin							&  100\%	&  100\%	&  100\%		&  0.7	&  0.1	&  0.1	&  1	\\
atlanta							&  0\%	&  $\geq$4.6\%	&  $\geq$100\%		&  0.4	&  1164.2	&  (3.2\%)	&  6	\\
dfn-gwin							&  100\%	&  100\%	&  101\%		&  0.4	&  0.3	&  0.3	&  1	\\
sun							&  69\%	&  93\%	&  325\%		&  $<$ 0.1 	&  3.3	&  59.4	&  1498	\\
newyork							&  14\%	&  89\%	&  1284\%		&  0.1	&  113.5	&  580.7	&  1227	\\
ta1							&  100\%	&  100\%	&  100\%		&  1.5	&  661.7	&  661.7	&  1	\\

\hline
\end{tabular}
} 
} 
\end{subtable}

} 
\caption{Performances of the $\mathcal{TIGHT-INNER}$ and $\mathcal{PATTERN}$ relaxations of the Unsplittable-CMCF problem, and branch-and-price performances using the $\mathcal{PATTERN}$ relaxation.}
\label{tab:gap_and_perf}
\end{table}

\FloatBarrier 
\section{Conclusion}

In this work, we presented a generalization of both the Splittable and Unsplittable Multi-Commodity Flow problems, in which the optimized cost is a sum of convex functions of the flow over the graph arcs. 
This problem can be used, for instance, in Telecommunication to model the delay of network devices. It also provides routings that avoid high levels of congestion on the network link and leave free capacity to route unknown future demands.

We showed that the Splittable problem can be solved in polynomial time for a given precision and that the Unpslittable one was $\mathcal{NP}-hard$. 
We also presented three different formulation of the Splittable problem, a compact formulation $\mathcal{COMPACT}$, a Simplicial Decomposition $\mathcal{CONVEX}$ and an inner-approximation $\mathcal{INNER}$. Experiments showed the superiority of the $\mathcal{INNER}$ approximation method in terms of computational performance, with solving duration divided by $10$ to $35$ on all large instances of the SND library. We also demonstrated the higher flexibility of the $\mathcal{INNER}$ method, allowing it to optimize routings even with black-box cost functions.

Finally, we considered the Unsplittable-CMCF Problem and showed how it can be solved with a branch-and-bound method, by leveraging two tightening of the $\mathcal{INNER}$ formulation, one directly adding constraints to the inner-approximation called $\mathcal{TIGHT-INNER}$, relying on the unsplittability of each flow, and one using a new set of variables representing commodity patterns (called $\mathcal{PATTERN}$), leveraging the unsplittability hypothesis at the level of each arc.
The computational results show that the tighter pattern approach allows us to solve most of the SNDLib instances tested.

This work can be applied to minimize the average delay in a telecommunication network, where each delay is a Kleinrock function of its utilization. Further work will be conducted to minimize the maximum delay over a set of demand, which would require, even for a Splittable problem, to solve non-convex problems.

\bibliographystyle{plain}
\bibliography{biblio}

\end{document}